\documentclass[sigconf,natbib=true]{acmart}
\setcopyright{None}
\renewcommand\footnotetextcopyrightpermission[1]{} 
\acmConference[WWW '22]{Proceedings of The Web Conference 2022 (WWW '22)}{Apr 25--29, 2022}{Lyon, France}
\acmBooktitle{Proceedings of The Web Conference 2022 (WWW '22), Apr 25--29, 2022, Lyon, France}

\usepackage[ruled,vlined]{algorithm2e}
\usepackage{multirow}
\usepackage{dutchcal}
\usepackage{subfigure}  
\usepackage{amsmath}
\newtheorem{proposition}{Proposition}
\usepackage{graphics}

\begin{document}

\title{COREATTACK: Breaking Up the Core Structure of Graphs}

\author{Bo Zhou}
\authornote{Both authors contributed equally to this research.}
\email{zb@zjvtit.edu.cn}
\affiliation{
  \institution{Zhejiang University of Technology}
  \city{Hangzhou}
  \state{Zhejiang}
  \country{China}
  \postcode{310023}
}

\author{Yuqian Lv}
\author{Jinhuan Wang}
\author{Jian Zhang}
\affiliation{
  \institution{Zhejiang University of Technology}
  \city{Hangzhou}
  \state{Zhejiang}
  \country{Chinal}}
  
\author{Qi Xuan}
\authornotemark[1]
\email{xuanqi@zjut.edu.cn}
\affiliation{
  \institution{Zhejiang University of Technology}
  \city{Hangzhou}
  \state{Zhejiang}
  \country{China}
  \postcode{310023}
}








\balance
\begin{abstract}


The concept of \emph{k}-core in complex networks plays a key role in many applications, e.g., understanding the global structure, or identifying central/critical nodes, of a network. A malicious attacker with jamming ability can exploit the vulnerability of the \emph{k}-core structure to attack the network and invalidate the network analysis methods, e.g., reducing the \emph{k}-shell values of nodes can deceive graph algorithms, leading to the wrong decisions. In this paper, we investigate the robustness of the \emph{k}-core structure under adversarial attacks by deleting edges, for the first time. Firstly, we give the general definition of targeted \emph{k}-core attack, map it to the set cover problem which is NP-hard, and further introduce a series of evaluation metrics to measure the performance of attack methods. Then, we propose $Q$ index theoretically as the probability that the terminal node of an edge does not belong to the innermost core, which is further used to guide the design of our heuristic attack methods, namely COREATTACK and GreedyCOREATTACK. The experiments on a variety of real-world networks demonstrate that our methods behave much better than a series of baselines, in terms of much smaller Edge Change Rate (ECR) and False Attack Rate (FAR), achieving state-of-the-art attack performance. More impressively, for certain real-world networks, only deleting one edge from the \emph{k}-core may lead to the collapse of the innermost core, even if this core contains dozens of nodes. Such a phenomenon indicates that the $k$-core structure could be extremely vulnerable under adversarial attacks, and its robustness thus should be carefully addressed to ensure the security of many graph algorithms. 

\end{abstract}




\maketitle

\section{Introduction}

Networks or graphs well represent various complex systems in our daily lives. Over the past few decades, many graph algorithms have been developed to gain knowledge from the structural properties of these networks~\cite{xuan2021graph}. One of the most widely used methods is called the \emph{k}-core decomposition. The \emph{k}-core~\cite{seidman1983network} of a graph, defined as the maximal subgraph such that every node has at least \emph{k} neighbors, have emerged as an important concept for understanding the global structure of networks, as well as for identifying \emph{central} nodes within a network. In fact, \emph{k}-cores represent cohesive subgroups of nodes, and have been used in a broad variety of important applications. In biology networks~\cite{PhysRevE.82.051911}, they were used to predict the feature of functional-unknown proteins~\cite{altaf2003prediction}, or understand the evolution of networks~\cite{wuchty2005peeling}. In ecology networks~\cite{garcia2017ranking,filho2018hierarchical}, they were used to construct a metabolic network of 17 plants covering unicellular organisms. In social networks~\cite{PhysRevE.70.056122,Cha_Haddadi_Benevenuto_Gummadi_2010}, they were adopted to identify the key nodes in the network, so as to measure the influence of users. In information spreading~\cite{PhysRevE.85.026116,miorandi2010k}, they were further extended to determine the impact of nodes in epidemics. Besides, the \emph{k}-core decomposition is also adopted in community detection~\cite{guimera2003self,giatsidis2011evaluating} as a method of assessing the nature of community collaboration. 
Since so many applications depend on the \emph{k}-core structure of a network, it is crucial to address the robustness of \emph{k}-core structure under adversarial attacks. In fact, adversarial attacks have been carefully investigated for various graph algorithms~\cite{shan2021adversarial}, such as node classification~\cite{bhagat2011node}, link prediction~\cite{9531428}, community detection~\cite{fortunato2010community}, and graph classification~\cite{zhang2018end}, to address their robustness under tiny perturbation on network structure. However, few studies focus on breaking up the core structure of a graph. 


\begin{figure}[ht]
  \centering
  \includegraphics[width=0.7\linewidth]{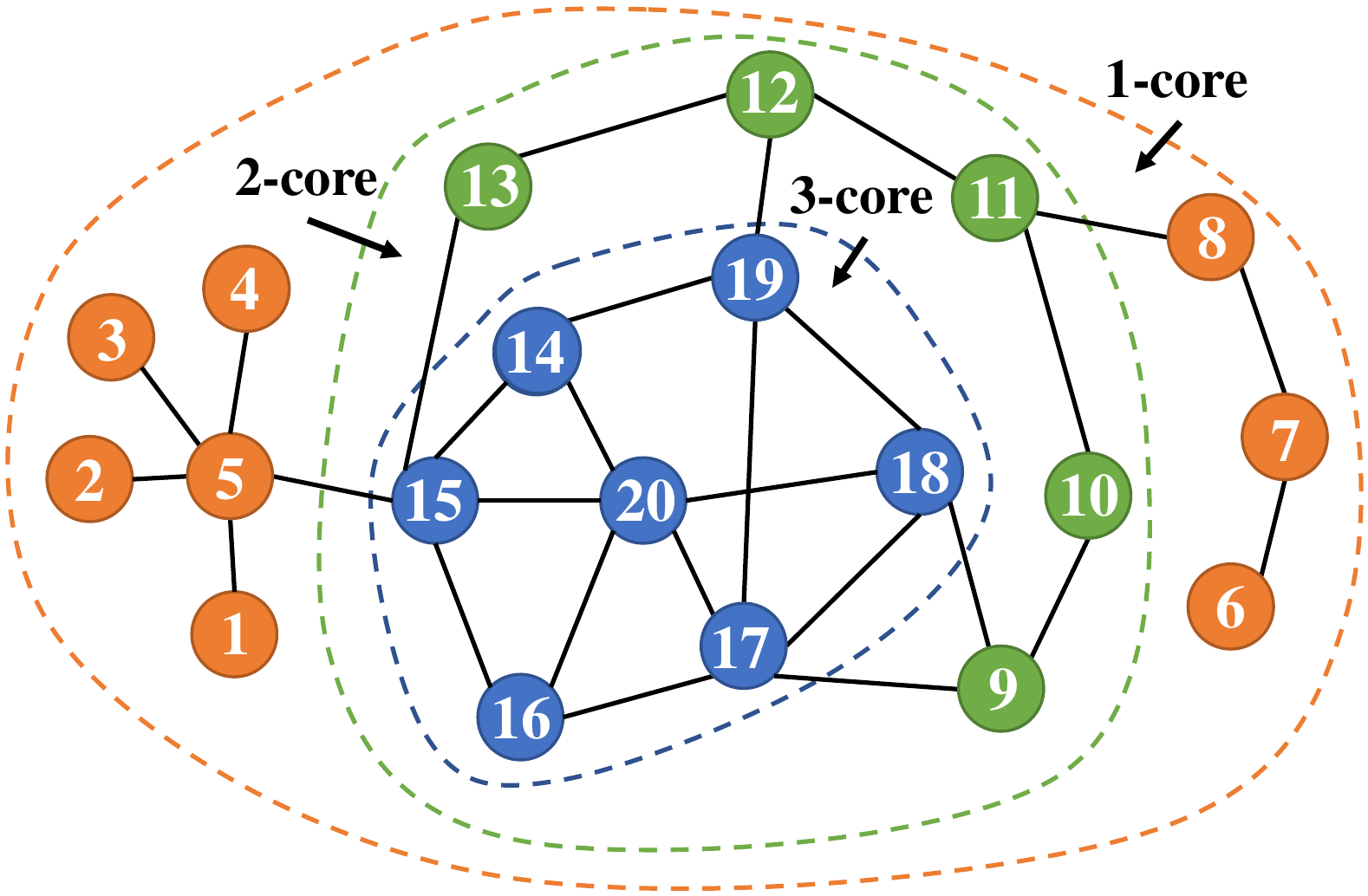}
  \caption{An example of the \emph{k}-core decomposition. The nodes within the dotted line of certain color form the \emph{k}-core of the graph, with $k$ varying from 1 to 3 here. And the nodes marked by the same color represent certain \emph{k}-shell.
  }
\label{figure:kcore}
\end{figure}

Generally, the \emph{k}-core of a graph $G$ can be obtained by recursively removing the nodes whose degrees are less than \emph{k}, with time complexity $\mathcal{O}(m)$~\cite{batagelj2003m}, where $m$ is the number of edges. 
Typically, the nodes in the higher cores are considered to be more central within the network. To test the robustness of the $k$-core structure, it is interesting to see at least how many links need to be changed to make the highest \emph{k}-core of the graph collapse. Taking the graph in Figure~\ref{figure:kcore} as an example, deleting edge $e_{17,19}$, but not $e_{17,20}$, will destroy the highest \emph{k}-core structure of the graph, i.e., the \emph{3}-core subgraph vanishes, although edge $e_{17,20}$ has the terminal nodes of larger degree and thus could be considered more crucial in traditional network analysis. This simple example indicates that the \emph{k}-core structure of a graph could be extremely vulnerable under adversarial attack when the appropriate edges are selected to change.


In view of this, we propose a heuristic algorithm, called COREATTACK, to determine which edges should be removed from a network to make the highest \emph{k}-core structure collapse. Generally, one can recursively execute the algorithm to attack any \emph{k}-core structure of a graph. Extensive experiments on 15 real-world networks show the effectiveness of our method. Especially, our contributions can be summarized as follows:
\begin{itemize}
    \item We formulate and study the problem of adversarial attack on \emph{k}-core structure of a graph by removing edges for the first time, which is a typical NP-hard problem. 
    \item We present a heuristic method, namely COREATTACK, to attack the innermost \emph{k}-core structure of a graph, with its effectiveness theoretically proved.
    
    \item We testify COREATTACK on a variety of real-world networks and find it outperforms a set of baseline methods, validating its state-of-the-art performance. 
    
\end{itemize}

The rest of this paper is structured as follows. In Section~\ref{sec:related}, we make a brief review on the previous studies of \emph{k}-core. In Section~\ref{sec:problem}, we give the concepts and problem definition used in this paper, with the proof that targeted \emph{k}-core attack is NP-hard. In Section~\ref{sec:method}, we introduce our COREATTACK and prove its effectiveness in theory. In Section~\ref{sec:exp}, we give evaluation metrics and verify COREATTACK on 15 real-world networks, compared with several baselines. Finally, we conclude the paper in Section~\ref{sec:conclusion}.

\section{RELATED WORKS}\label{sec:related}
In this section, we describe some related works on \emph{k}-core decomposition, \emph{k}-core percolation, and \emph{k}-core robustness.

\textbf{Core decomposition:} Erd{\"o}s and Hajnal~\cite{erdHos1966chromatic} gave the first \emph{k}-core related concept in 1966, defining the degeneracy of a graph as the maximum core number of a node in the graph. 
Almost simultaneously, Seidman~\cite{seidman1983network}, as well as Matula and Beck~\cite{matula1983smallest}, defined the \emph{k}-core subgraph as the maximal connected subgraph where each node has at least degree \emph{k}. Seidman stated that \emph{k}-cores are good seedbeds that can be used to find further dense substructures.
Thanks to the practical benefit and linear complexity of the \emph{k}-core decomposition, there has been a great deal of recent work in adapting \emph{k}-core algorithms for different data types or setups. Cheng \emph{et al}.~\cite{5767911} introduced the first external-memory
algorithm; Wen \emph{et al}.~\cite{wen2016efficient} and Khaouid \emph{et al}.~\cite{khaouid2015k} provided
further improvements in this direction. Giatsidis \emph{et al}. adapted the \emph{k}-core decomposition for weighted~\cite{giatsidis2013d} and directed~\cite{giatsidis2011evaluating} graphs.


\textbf{Core percolation:} In the field of studying the dynamic characteristics of \emph{k}-core, Baxter \emph{et al.}~\cite{baxter2015critical} analyzed the critical dynamics of the \emph{k}-core pruning process and showed that the pruning process exhibits three different behaviors depending on whether the mean degree $q$ of the initial network is above, equal to, or below the threshold $q_c$ corresponding to the emergence of the giant \emph{k}-core. In addition, Goltsev \emph{et al.}~\cite{goltsev2006k} developed the theory of the \emph{k}-core (bootstrap) percolation on uncorrelated random networks with arbitrary degree distributions and showed that a random removal of even one node from the \emph{k}-core may result in the collapse of a vast region of the \emph{k}-core around the removed node.

\textbf{Core robustness:} Recently, Zhou \emph{et al.}~\cite{zhou2021robustness} studied the robustness of \emph{k}-shell structure under adversarial attacks by proposing simulated annealing based \emph{k}-shell attack method and gave the results that the \emph{k}-shell structure of a network is robust under random perturbation, but is quite vulnerable under adversarial attack through simulations on several real-world networks. 
Zhang \emph{et al.}~\cite{zhang2017finding} investigated the collapsed \emph{k}-core problem to find the critical nodes. For a given \emph{k} value and a budget \emph{b}, they introduced an algorithms to delete \emph{b} (critical) nodes to get the smallest \emph{k}-core in size. Besides, Schmidt \emph{et al.}~\cite{schmidt2019minimal} recently investigated the minimal contagious set problem to find the smallest set of nodes in a network whose removal results in an empty \emph{k}-core. 
The algorithm of~\cite{schmidt2019minimal} runs by implementing the following two steps iteratively: first the \emph{k}-core of graph is obtained and second the node of the largest degree in the \emph{k}-core is removed.

Although some previous works may cause a small collapse of \emph{k}-core by deleting nodes, their main purpose is to study the characteristics of \emph{k}-core decomposition or percolation. Moreover, deleting nodes will generally damage the structure of the network much more significantly than deleting edges, making it less precise and practical. To the best of our knowledge, we are the first to propose the concept of targeted \emph{k}-core attack by deleting edges.


\section{problem statement} \label{sec:problem}
Now, we introduce the preliminaries used in this paper and give the problem definition, meanwhile, give the proof that targeted \emph{k}-core attack is an NP-hard problem. In the end, we propose $Q$ index to measure attack strategy and give numerical simulation.

\subsection{Preliminaries}
A network or graph (we will use the two terms indiscriminately) is denoted by $G\left( V,E\right) $, where $V$ and $E\subseteq(V\times V)$ are the sets of nodes and edges respectively, which represents the real-world entities and the relationship between them. In this paper, we only focus on undirected graphs without self-loops or weights of edges. Here, we give introductions of the basic concepts mentioned and the definitions used later. Here, we give a summary of notations in Table~\ref{tab:freq2} for convenience. 

\begin{table}[htbp]
  \caption{Summary of notations.}
  \label{tab:freq2}
  \begin{tabular*}{\hsize}{@{}@{\extracolsep{\fill}}lr@{}}
    \toprule[0.5mm]
    Notation    &Definition\\
    \midrule
    $\mathbcal{B}\left(v,G\right)$  &the adjacent nodes of $v$ in $G$      \\
    $k_{max}$ & the largest \emph{k}-shell value\\
    $G_I$ &innermost core subgraph of $G$\\
    $G_k$ &\emph{k}-core subgraph of $G$\\
    $\mathbcal{A}\left(G\right)$      &all node in the graph $G$     \\
    $E\left(G\right)$ & all edges in the graph $G$ \\
    $C_k\left(G\right)$     & the corona subgraph of \emph{k}-core      \\
    $C_I\left(G\right)$ &the corona subgraph of innermost core graph       \\
    $\Gamma\left(e\right)$       &collapsed nodes caused by deleting edges $e$     \\
    $\left|G\right|$         &the number of nodes in $G$      \\
    $\mathbcal{C}\left(v,G\right)$        &subgraph of node $v$ and its neighbors in graph $G$    \\
    $\Phi_{k}(\cdot)$ & extracting the \emph{k}-core subgraph from graph $G$\\
    %
    \bottomrule[0.5mm]
    \end{tabular*}
\end{table}


\textbf{\emph{k}-core subgraph:} The \emph{k}-core of a graph $G$ is its largest subgraph whose nodes have degree at least $k$~\cite{seidman1983network}, which is denoted as $G_k(V_k,E_k)$, satisfying $V_k \subset V$ and $E_k \subset E$. In other words, each of nodes in the \emph{k}-core has at least $k$ neighbours within this subgraph. 
The core number of a node is defined to be the largest \emph{k} value, i.e., there exists a \emph{k}-core (but we cannot find a $k'$-core satisfying $k'>k$) that contains the node. As shown in Figure~\ref{figure:kcore}, \emph{1}-core is the graph $G$, \emph{2}-core is the subgraph circled by the green dotted line, and \emph{3}-core is the subgraph circled by the blue dotted line.

\textbf{\emph{k}-shell subgraph:} The \emph{k}-shell is the subgraph of $G$, with its nodes belonging to the \emph{k}-core but not to the $(k+1)$-core~\cite{malliaros2020core}. Each node in the graph will have its own value of \emph{k}-shell which can be utilized to measure the importance of the node in the graph. Every node belongs to a unique \emph{k}-shell subgraph. As shown in Figure~\ref{figure:kcore}, the orange nodes belong to \emph{1}-shell, the green nodes belong to \emph{2}-shell, and the blue nodes belong to \emph{3}-shell.


\textbf{Innermost core:} The innermost core is the $k$-core with the largest \emph{k} value. For the graph in Figure~\ref{figure:kcore}, the largest \emph{k} value is 3, and thus the innermost core of this graph is \emph{3}-core. For convenience, we denote the innermost core as $G_I=(V_I,E_I)$, satisfying  $V_I \subset V_k$ and $E_I \subset E_k$ when $I>k$.

\begin{figure}[ht]
  \centering
  \includegraphics[width=0.8\linewidth]{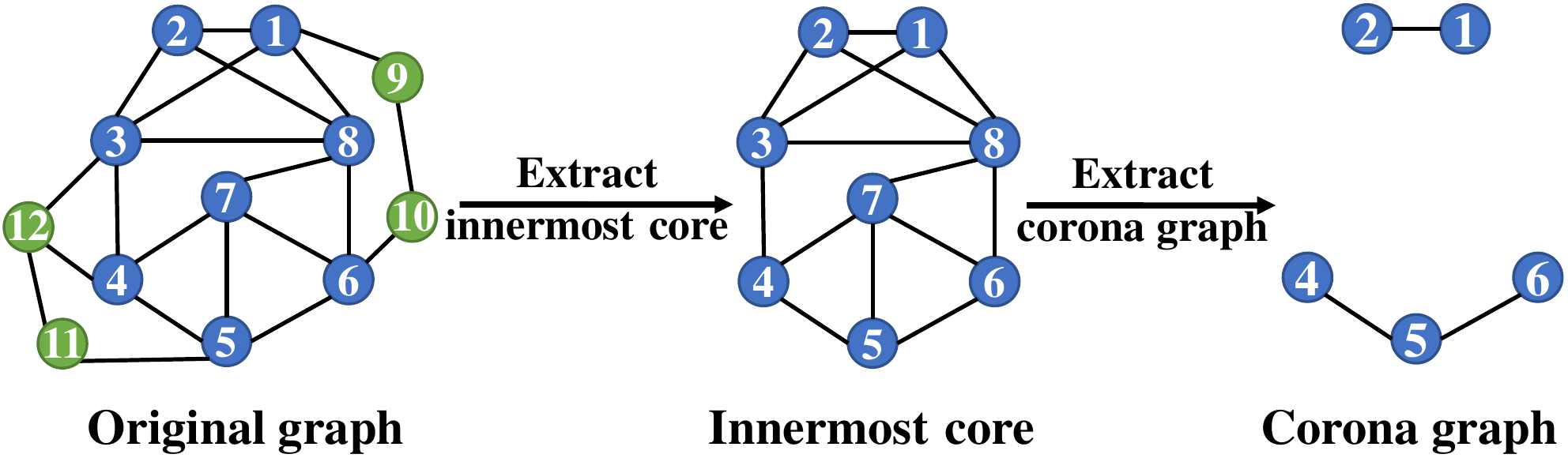}
  \caption{The process of extracting the innermost core and the corona graph from an original graph.}
\label{figure:corona}
\end{figure}

\textbf{Corona graph:} Given an innermost core of a graph $G_I=(V_I,E_I)$, then all the nodes within it have degree equal or larger than $I$. In this case, the corona graph is a subgraph of innermost core, where all the nodes have exactly \emph{I} neighbours~\cite{dorogovtsev2006k}. It should be noted that corona graph may be unconnected in reality, especially for those large networks. Figure~\ref{figure:corona} gives an example to extract the innermost core and corona graph from a simple graph. As we can see, the corona graph has two disconnected components: $\left\{1,2\right\}$ and $\left\{4,5,6\right\}$.

\subsection{Problem Definition}
Since innermost core plays a significant role in information transmission of the whole network, in this paper, we focus on effectively destroying the innermost core of a graph by deleting a set of appropriately selected edges.

\textbf{Targeted \emph{k}-core attack.} An adversary with full knowledge of $G$ seeks to precisely decrease the \emph{k}-shell value of the nodes in the innermost core, but keep the \emph{k}-shell values of the rest nodes the same, by removing as few edges as possible, which is formulated by
\begin{align*}
     &\mathop{\min} |\Delta E| \\
     s.t. &\left|\Phi_{I}(G \setminus \Delta E)\right| = 0
     \label{eq:Def}
\end{align*}
where $G$ is the original graph, $\Delta{E}$ represents the set of deleted edges, $\Phi_{I}(\cdot)$ is the operator to obtain the innermost core of $G$. $I$ is the larges \emph{k}-shell value of the original graph $G$.
Note that removing fewer edges corresponds to a lower cost to realize the attack, which will also make the attack less detectable based on network structure. 
\textbf{Gainer set of edge.} 
For a given innermost core of $G$, there are three types of edges in $G$: (i) edges connecting two nodes inside the innermost core; (ii) edges connecting two nodes outside the innermost core; (iii) edges connecting a node inside the innermost core and another outside the innermost core. Based on the \emph{k}-core decomposition method~\cite{seidman1983network}, both case (ii) and case (iii) do not change the structure of the innermost core, as a result, here we only need to consider case (i).
Suppose we delete edge $e \in E_I$ from the innermost core, a set of nodes $\Gamma(e) \subset V_I$ may not belong to innermost core any more. We name these nodes in $\Gamma(e)$ as the gainers of edge $e$, which is defined as
\begin{equation}
    \Gamma(e)=\Phi_I\left(G(V,E)\right)\setminus\Phi_I\left(G(V,E\setminus e)\right).
\end{equation}
As shown in Figure~\ref{figure:coreAttack}, the gainer set of edge $e_{1,2}$ is $\Gamma(e_{1,2}) = \{v_1,\cdots,v_8\}$, while that of edge $e_{7,8}$ is $\Gamma(e_{7,8})=\emptyset$. It's obvious that edge $e_{1,2}$ is much more important than $e_{7,8}$ to maintain the innermost core $G_I$, since $|\Gamma(e_{1,2})|=8 \gg |\Gamma(e_{7,8})|=0$.

\begin{figure}[h]
  \centering
  \includegraphics[width=.9\linewidth]{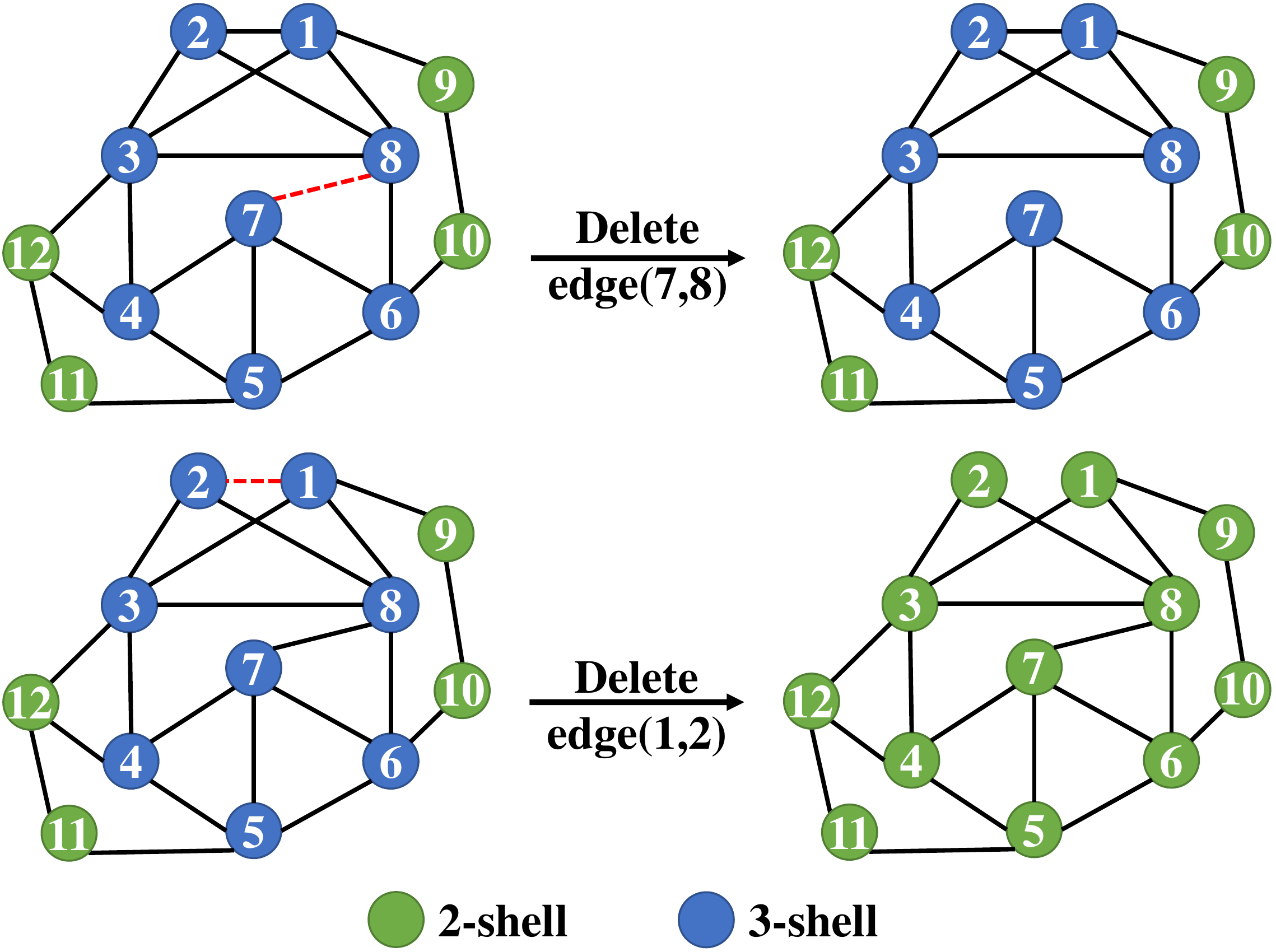}
  \caption{Given a graph $G$ with nodes $\left\{9,10,11,12\right\}$ belonging to \emph{2}-shell and nodes $\left\{1,2,3,4,5,6,7,8\right\}$ belonging to \emph{3}-shell (left), the \emph{3}-core doesn't collapse when edge $e_{7,8}$ is removed (top right); while it indeed totally collapses into \emph{2}-core when edge $e_{1,2}$ is deleted (bottom right).}
\label{figure:coreAttack}
\end{figure}

\textbf{NP-hard problem.} Mathematically, in \emph{k}-core attack, we try to use the minimum number of gainer sets $\Gamma(e), e\in E_I$, to cover the node set $V_I$ in the innermost core $G_I$, which is equivalent to the classic Set Cover problem~\cite{feige1998threshold}. In this sense, \emph{k}-core attack is a typical NP-hard problem~\cite{cormen2009introduction}. In particular, given the innermost core of a graph, i.e., $G_I=(V_I,E_I)$, suppose we need to delete at least $d$ edges to make the core totally collapse, in experiment, we thus need to enumerate all edge combinations from 1 to $d$ exhaustively and check their effects. The number of all possible combinations can be calculated by
\begin{equation}
    f = \sum_{i=1}^{d}C_{|E_I|}^{i}.
	\label{eqNP}
\end{equation}
Based on Eq.~(\ref{eqNP}), we can see that when the number of edges in the innermost core is in the thousands, the computational complexity could be extremely high even if $d$ is only in the tens. Therefore, suitable heuristic strategies are yearned for finding an effective edge combination to realize the attack. 

\subsection{Theoretical Analysis}

Generally, there are two possible ways to disintegrate the \emph{k}-core of a network: removing nodes or deleting edges. Removing a node from a network always involves a bunch of relating edges, which might make the attack higher cost and more noticeable. Thus, we consider edge deletion here as the way to achieve targeted \emph{k}-core attack. There are two types of deleted edges in the innermost core: (i) neither of the two terminal node of the deleted edge belong to corona subgraph; (ii) at least one terminal node of the deleted edge belongs to the corona subgraph.

Considering an uncorrelated $G=(V,E)$, the aim of \emph{k}-core attack is to disintegrate the \emph{k}-core of $G$ by deleting $L$ edges. For a randomly selected edge, its probability of not being deleted is $p = 1-\frac{L}{\left|E\right|}$, $P(i)$ is the degree distribution of graph $G$. And $z_1$ is the average degree of $G$, which is defined as 
\begin{equation}
    \begin{aligned}
        z_1=\sum_{i}^{\infty}iP(i)=\frac{2\left|E\right|}{\left|V\right|}.
    \end{aligned}
\end{equation}
If the edge is not deleted, follow the edge to a node of degree $i$, in which the probability of $l$ adjacent edges being deleted is calculated by
\begin{equation}
    \begin{aligned}
        \mathbcal{Q}(i,l)=\frac{i P(i)}{z_1}\binom{i-1}{l}(1-p)^l p^{i-l}.
    \end{aligned}
\end{equation}
If the edge is deleted, then the probability is
\begin{equation}
    \begin{aligned}
        \mathbcal{R}(i,l)=\frac{i P(i)}{z_1}\binom{i-1}{l-1}(1-p)^{l}p^{i-l}.
    \end{aligned}
\end{equation}
Let $Q$ the probability that a terminal node of an edge in graph $G$ does not belong to innermost core, when we randomly select $L$ edges from $G$ to be removed, which can then be calculated by
\begin{equation}
\begin{aligned}
    Q = &\sum_{n=0}^{k-2}\sum_{l=0}^{L}\sum_{i=n+l}^{\infty}[\mathbcal{Q}(i+1,l)\binom{i-l}{n}Q^{i-n-l}(1-Q)^n]\\
    &+\sum_{n=0}^{k-1}\sum_{l=0}^{L-1}\sum_{i=n+l}^{\infty}[\mathbcal{R}(i+1,l+1)\binom{i-l}{n}Q^{i-n-l}(1-Q)^n].
\end{aligned}
\end{equation}

\begin{proposition}
    Given a network $G$ and its k-core $G_k$, a node $v \in G$ belongs to $G_k$ if and only if $\left| \mathcal{B}(v, G) \cap \mathcal{B}(v, G_k) \right| \geq k$.
\end{proposition}

\begin{proof}
For a node $ v \in G$, $\left| \mathcal{B}(v, G) \cap \mathcal{B}(v, G_k) \right| > k$ implies that $|\mathcal{B}(v, G_k)|> k$. Thus $v$ belongs to $G_k$. If $v \in G_k$, then $|\mathcal{B}(v, G_k)| > k$. Since $G_k \subset G$, we have $\left| \mathcal{B}(v, G) \cap \mathcal{B}(v, G_k) \right| \geq k$.
\end{proof}

The proposition indicates that we should focus on the nodes that exactly have $k$ neighbors in $G_k$, i.e. the nodes in the corona. And we conduct simulations on Erd{\"o}s-R{\'e}nyi (ER) random graphs to prove the hypothesis. In an ER random graph, we consider edge-deletion case (i) and case (ii) to explore their influence on $Q$. Under these two cases, our numerical simulations are given the same number of deleted edges. It can be seen from the Figure~\ref{figure:ER_Q} that deleting the edge satisfying case (i) has little effect on the $Q$ index of the whole network, while deleting the edge satisfying case (ii) will cause the innermost core sudden disappearance. This phenomenon in physics is called percolation, which means that the innermost core immediately collapse when certain number of edges are removed.

In fact, when only considering the edges relating to the corona in edge deletion, there always exists a mutation after a few edges are removed and then the \emph{k}-core collapsed. It strongly implies the significant influence that such edges have on the structure of \emph{k}-core, which could be used in the design of \emph{k}-core attack.

\begin{figure}[t]
    \centering
    \subfigure{\includegraphics[width=.49\linewidth]{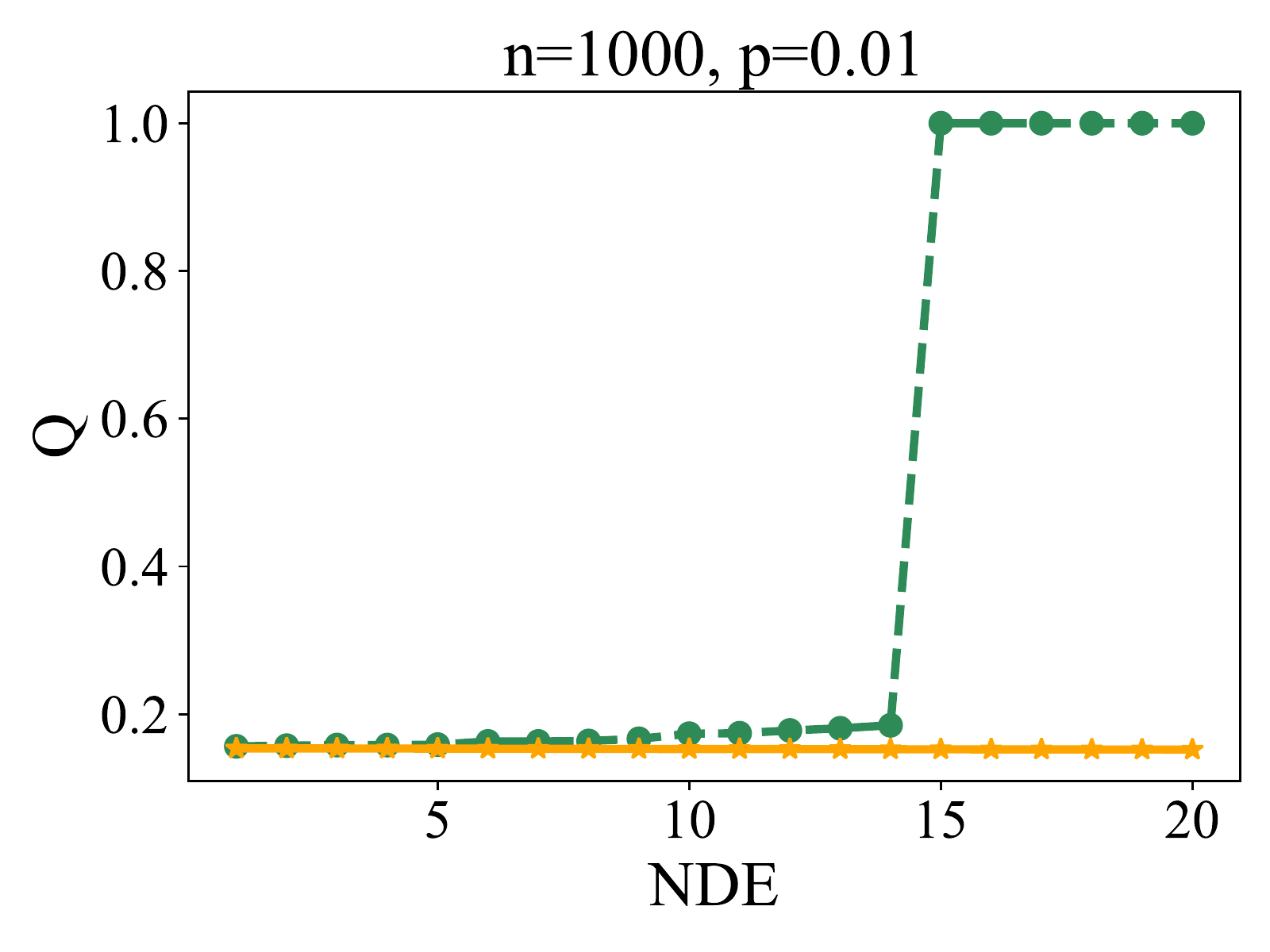}\label{4.a}}
    \subfigure{\includegraphics[width=.49\linewidth]{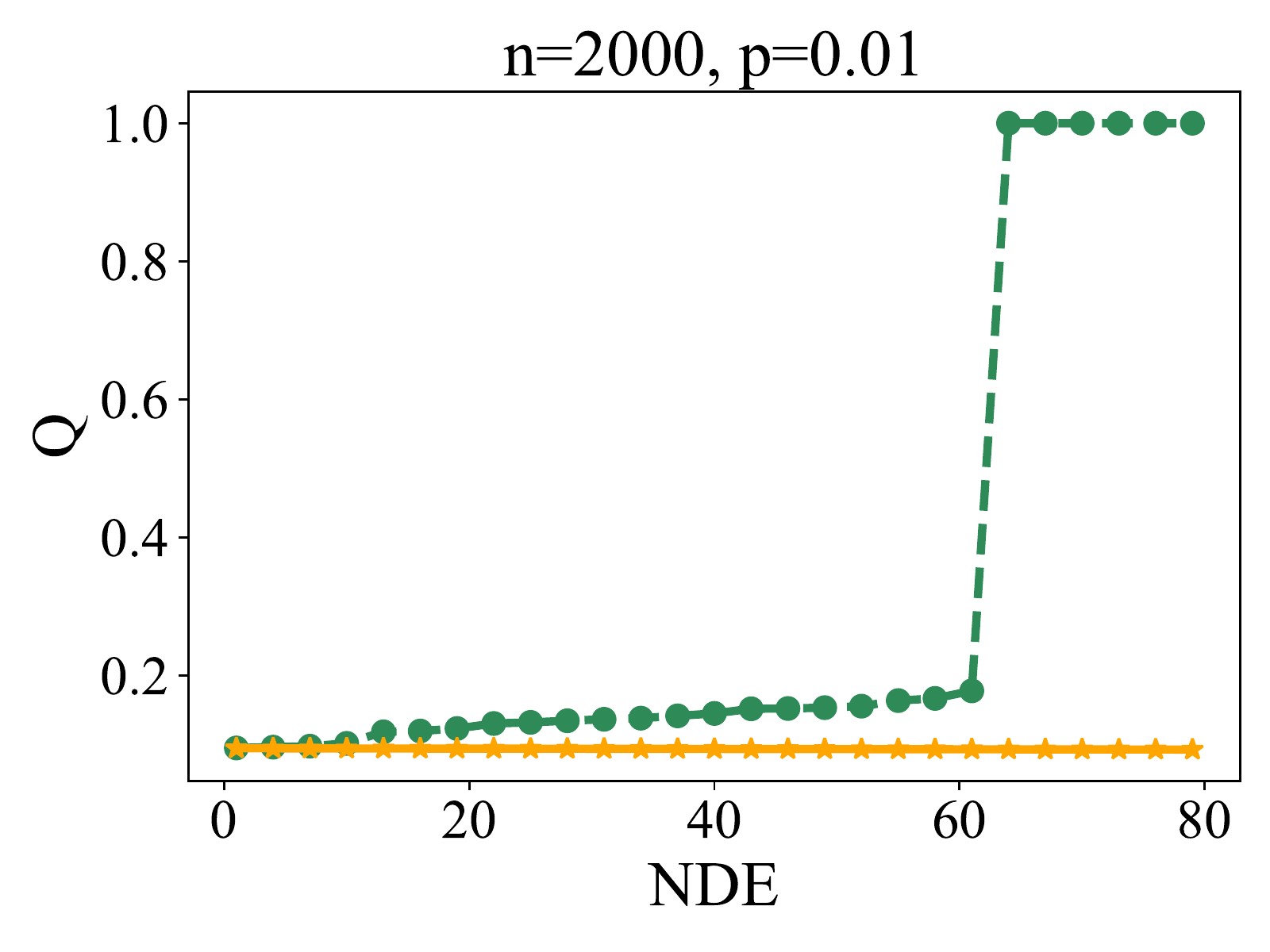}\label{4.b}}
    \subfigure{\includegraphics[width=.8\linewidth]{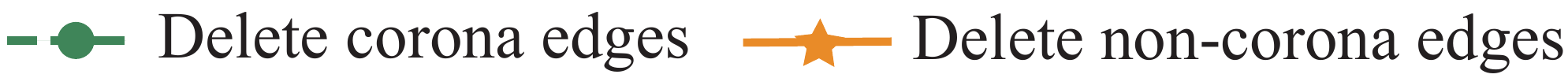}}
    \vspace{-4mm}
    \caption{The influence of different edge deletion strategies on $Q$ index with different ER random graphs, where $x$-axis is the number of deleted edges, $\text{NDE}=L$.}
    \label{figure:ER_Q}
\end{figure}

\section{methodology}\label{sec:method}


In this section, we propose a heuristic attack strategy, namely COREATTACK, to make the \emph{k}-core of a network collapse. And an optimized strategy based on maximum set cover is also presented to reduce the attack cost.

\begin{figure*}[htbp]
  \centering
  \includegraphics[width=0.8\linewidth]{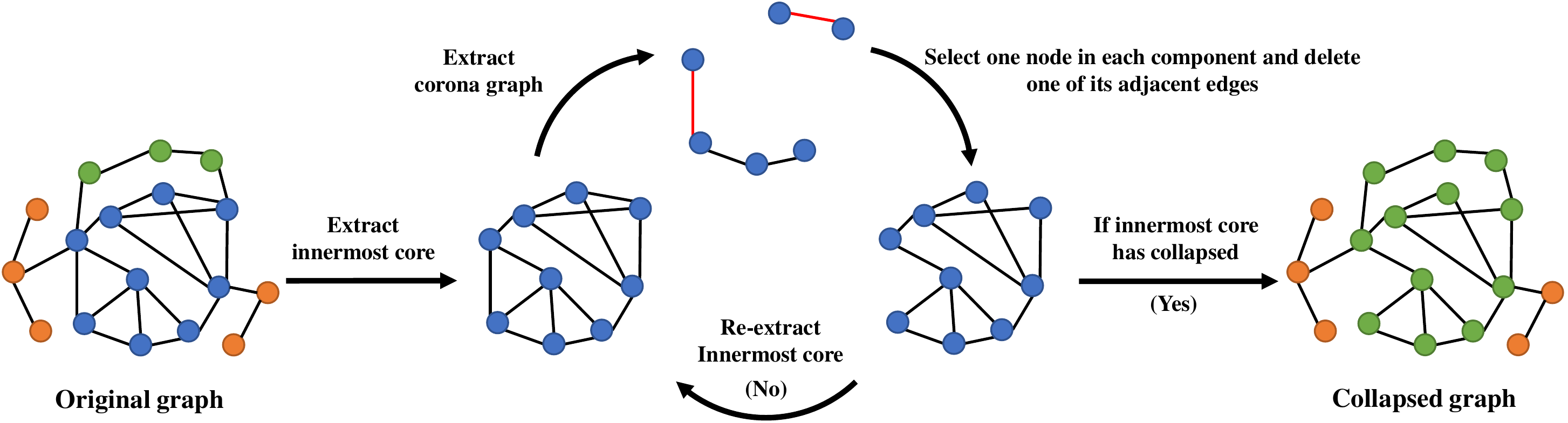}
  \caption{The framework of COREATTACK. We extract the innermost core from the original graph, and then the corona graph from the innermost core, where we get two components in this example. In each component, we select a node and delete one of its adjacent edges within innermost core. Afterwards, if the innermost core collapse after the attack, we output the collapsed graph, otherwise, we re-extract the innermost core from the attacked graph and repeat the above process.}
\label{figure:process}
\end{figure*}

\subsection{COREATTACK}
Different edges play different roles in a network and may have different effects on the collapse of the \emph{k}-core. As stated in Section~\ref{sec:problem}, we define such effect as the Gainer set of an edge. It is an intuitive yet comprehensible idea that the edges with larger Gainer set lie in the corona of a given $G_I$.

As shown in Figure~\ref{figure:coreAttack}, the terminal nodes of $e_{1,2}$ belong to the corona graph, making the edge an effective one for the collapse of the \emph{k}-core. However, the deletion of $e_{7,8}$ whose endpoints do not belong to the corona could not lead to the collapse of innermost core. In the meanwhile, we notice that the gainer sets of different edges might overlap with each other, in other words, there might be a case that an edge will be automatically removed from $G_I$ if another edge is deleted. Considering this, we have a proposition here to guide the edge deletion.

\begin{proposition}
    Given an innermost core $G_I \subset G$ and an edge $e_{i,j}\subset E_I$, we have $\Gamma(e_{u,v}) \subset \Gamma(e_{i,j})$ for $ \forall e_{u,v} \in E(\mathbcal{C}(\Omega, G_I))$ where $\Omega=\mathcal{A}(C_I(G)) \cap \Gamma(e_{i,j})$ and $\Omega \neq \emptyset$.
\end{proposition}

\begin{proof}
The deletion of $e_{i,j}$ could directly lead to the nodes in $\Omega$ being remove from $G_I$ by \emph{k}-core decomposition. For $\forall e_{u,v} \in E(\mathcal{C}(\Omega, G_I))$, it as well as its gainer set will also be removed from $G_I$ if $e_{i,j}$ is deleted, which implies $\Gamma(e_{u,v}) \subset \Gamma(e_{i,j})$.
\end{proof}



According to the proposition, we propose the first edge deletion strategy, \textbf{COREATTACK}. As described in Figure~\ref{figure:process}, COREATTACK extracts the innermost \emph{k}-core $G_I$ for input graph $G$ as the target for attack. In each attack iteration, it first finds the corona $C_I(G_I)$ and then randomly selects an edge, denoted by $e_{i,j}$ with at least one endpoint belonging to $C_I(G_I)$ for deletion. Afterwards, it collects the gainer set of $e_{i,j}$ which is denoted by $\Gamma(e_{i,j})$ and then removes the nodes in the set as well as relating edges from $C_I(G_I)$. When all the nodes and edges of the corona are removed, we obtain an edge set $E'$ for deletion in this iteration. We then update $G_I$ with $E'$ and run the next iteration. The attack procedure ends until $G_I$ collapses. The method is summarized in Algorithm~\ref{alg:1}.

\begin{algorithm}[h]
\caption{\textbf{COREATTACK}}
\LinesNumbered
\label{alg:1}
\KwIn{the original graph $G$, the core value $I$ of the innermost core under graph}
\KwOut{the set of deleted edges $E^{'}$.}
$G_I$ = MaxKSubgraph$(G)$\;
$E^{'}=\emptyset$\;
\While {$\left|G_I\right| > 0$}
{
    $C_I(G_I)=$CorSubgraph$(G_I)$\;
    \While{$\left|C_I(G_{I})\right| > 0$}
    {
        $v_i=$RandomSelect$\left(\mathbcal{A} \left(C_I\left(G_I\right)\right)\right)$\;
        $v_j =$RandomSelect$\left(\mathbcal{B}(v_i,G_I)\right)$\;
        $e_{i,j}=\left(v_i,v_j\right)$\;
        
        $C_I\left(G_I\right) = C_I\left(G_I\right)\setminus \Gamma\left(e_{i,j}\right)$\;
        $E^{'}=E^{'}\cup \{ e_{i,j} \}$\;
        
    }
    $G_I=G_I \setminus E^{'}$\;
    $G_I=$MaxKSubgraph$\left(G_I\right)$\;
}
\Return $E^{'}$\;
\end{algorithm}

\subsection{GreedyCOREATTACK}
Focusing on the edges in $G_I$ could indeed reduce the attack cost. The COREATTACK method is to delete the set of edges that can completely cover the corona node every search loop. Intuitively, this method may cause more edges to be deleted. Therefore, we improve the edge deletion procedure and propose \textbf{GreedyCOREATTACK} based on \textbf{COREATTACK}.

The main difference between the two methods is that we delete the edge with maximum gainer set rather than a random one in each iteration. And we name the edge as \emph{bomb edge}. To find the edge in a given $G_I$, we first random select an edge with at least one node belongs to the current corona subgraph and calculate the size of gainer set of the edge. Then, based on proposition 1, we exclude the gainer set of the edge from the corona subgraph. Finally, we could obtain the \emph{bomb edge} which have the largest size of gainer set. And the procedure is presented in Algorithm~\ref{alg:2}.

\begin{algorithm}[h]
\caption{\textbf{GreedyCOREATTACK}}
\LinesNumbered
\label{alg:2}
\KwIn{the original graph $G$, the core value $I$ of the innermost core under graph}
\KwOut{the set of deleted edges $E^{'}$.}
$G_I =$MaxKSubgraph$\left(G\right)$\;
$E^{'}=\emptyset$\;
\While {$G_I \neq \emptyset$}
{
    $C_I(G_I)=$CorSubgraph$\left(G_I\right)$\;
    $\tau=0$\;
    \While{$\left|C_I(G_I)\right|!=0$}
    {
        $v_i=$RandomSelect$\left(\mathbcal{A} \left(C_I\left(G_I\right)\right)\right)$\;
        $v_j =$RandomSelect$\left(\mathbcal{B}\left(v_i,G_I\right)\right)$\;
        $e_{i,j}=\left(v_i,v_j\right)$\;
        
        \If{$\left|\Gamma\left(e_{i,j}\right)\right|>\tau$}
        {
            $\tau=\left|\Gamma\left(e_{i,j}\right)\right|$\;
            $e^{'}=e_{i,j}$\;
        }
        $C_I\left(G_I\right)=C_I\left(G_I\right)\setminus \Gamma\left(e_{i,j}\right)$\;
    }
    $E^{'}=E^{'}\cup e^{'}$\;
    $G_I=G_I\setminus E^{'}$\;
    $G_I=$MaxKSubgraph$\left(G_I\right)$\;
}
\Return $E^{'}$\;
\end{algorithm}

\section{experiments}\label{sec:exp}
Now, we perform the experiments on 15 real-world networks of various sizes to demonstrate the performance of our COREATTACK and GreedyCOREATTACK methods, comparing with two baselines, under the four performance metrics introduced in Section~\ref{EM}.
\subsection{Evaluation Metrics}
\label{EM}
We use \emph{Number of Deleted Nodes} (NDN), \emph{Number of Deleted Edges} (NDE), \emph{Edge Change Rate} (ECR), and \emph{False Attack Rate} (FAR) to measure the effectiveness of various attacks. Note that we don't use \emph{Attack Success Rate} (ASR) here, since we always seek to destroy the innermost core as a whole and thus ASR will be equal to 1 for all the attack methods.

\begin{itemize}
\item \textbf{NDN:} If a node $v_{i} \in V$ is deleted from $G$, we set $\phi_{i}=1$, otherwise $\phi_{i}=0$, then, NDN is defined as:
\begin{equation}
	\centering
	\text{NDN}=\sum_{v_{i}\in{V}}\phi_{i}.
	\label{eqNDN}
\end{equation}

\item \textbf{NDE:} If an edge $e_{i,j} \in E$ is deleted from $G$, we set $\psi_{i,j}=1$, otherwise $\psi_{i,j}=0$, then, NDE is defined as:
\begin{equation}
	\centering
	\text{NDE}=\sum_{e_{i,j}\in{E}}\psi_{i,j}.
	\label{eqNDE}
\end{equation}

\item \textbf{ECR:} Since real-world networks are of varied size, therefore NDE may not be a suitable metric for the relative perturbation introduced by the attack. In this case, ECR may be a better metric, which is defined as:
\begin{equation}
	\centering
	\text{ECR}=\frac{\textit{NDE}}{|E|},
	\label{eqECR}
\end{equation}
where $|E|$ denotes the total number of edges in $G$. 


\item \textbf{FAR:} If the \emph{k}-shell of node $v_i$ outside the innermost core, i.e., $v_i \in V \setminus V_I$, is changed after the attack, we set $h_{i}=1$, otherwise $h_{i}=0$, then, FAR is defined as:
\begin{equation}
    \centering
    \text{FAR}=\frac{\sum_{v_i\in{V \setminus V_I}} h_{i}}{|V \setminus V_I|},
    \label{eqCSR}
\end{equation}
\end{itemize}
where $|V \setminus V_I|$ denotes the total number of nodes outside $G_I$. Note that NDN, NDE, and ECR are used to measure the budget to realize the attack, on the other hand, FAR is used to measure how precise the attack is. Therefore, the lower of these values, the better of the attack methods.

\subsection{Datasets}
The basic properties of 15 real-world networks are presented in Table~\ref{tab:dataset}. The labels (marked at the top-right corner of the name) indicate that these datasets are collected from different platforms. \emph{Dolphin} is a social network that describes the frequent associations between bottlenose dolphins living off Doubtful Sound, New Zealand~\cite{lusseau2003bottlenose}. \emph{Deezer}, \emph{Enron-Email}, \emph{GitHub}, and \emph{Gowalla} are collected from \url{http://snap.stanford.edu}~\cite{snapnets}. \emph{Human}, \emph{Web}, \emph{DBLP}, and \emph{FoodWeb} are collected from \url{https://networkrepository.com}~\cite{nr}. \emph{Proteins}, \emph{Autonomous}, \emph{Brightkite}, \emph{Cora}, \emph{JUNG-Javax} and \emph{US-Power} are collected from \url{http://konect.cc/networks/}~\cite{kunegis2013konect}. We treat these networks as undirected and unweighted, where self-loops and multiple edges are not allowed either.


\begin{table}[thbp]
  \caption{Basic properties of 15 real-world networks, i.e., the number of nodes $\left|V\right|$, the number of edges $\left|E\right|$, the largest \emph{k}-shell value $k_{max}$, the number of nodes in the innermost core $\left|V_I\right|$, the number of edges in the innermost core $\left|E_I\right|$. }
  \label{tab:dataset}
  \begin{tabular*}{\hsize}{@{}@{\extracolsep{\fill}}lrrrrr@{}}
    \toprule[0.5mm]
    Network&$\left|V\right|$&$\left|E\right|$&$k_{\max}$&$\left|V_{I}\right|$&$\left|E_{I}\right|$\\
    \midrule
    $\rm Dolphin^{\ast}$ &62 &159 &4 &36 &109\\
    $\rm Deezer^{\dagger}$&28,281&92,752 &12 &71 &564\\
    $\text{Enron-Email}^{\dagger}$& 36,692& 183,831 &43 &275 &9,633\\
    $\rm GitHub^{\dagger}$ &37,700 &289,003 &34 &360 &10,488\\
    $\rm Gowalla^{\dagger}$ &196,591 &950,327 &51 &185 &7,030\\
    $\text{Human}^{\star}$ &113 &2,196 &28 &60 &1,111\\
    $\rm Web^{\star}$ &163,598 &1,747,269 &101 &102 &5,151\\
    $\rm DBLP^{\star}$ &12,591 &49,620 &12 &916 &9,840\\
    $\rm FoodWeb^{\star}$ &128 &2,106 &24 &73 &1,193\\
    $\text{Proteins}^{\circ}$ &3,133 &6,149 &6 &137 &627\\
    $\text{Autonomous}^{\circ}$ &22,963 &48,436 &25 &71 &1,355 \\
    $\text{Brightkite}^{\circ}$ &58,228 &214,078 &52 &154 &5,919\\
    $\text{Cora}^{\circ}$ &23,166 &89,157 &13 &25 &218\\
    $\text{JUNG-Javax}^{\circ}$ &6,120 &50,290 &65 &135 &5,145\\
    $\text{US-Power}^{\circ}$ &4,941 &6,594 &5 &12 &36\\ 
  \bottomrule[0.5mm]
\end{tabular*}
\end{table}

\subsection{Baselines}
We adopt the four attack strategies, including Random Edge Deletion (RED), Collapsed $K$-Core (CKC) based on node deletion~\cite{zhang2017finding}, removing Highest Degree Nodes (HDN)~\cite{schmidt2019minimal}, and removing the Highest Degree Edges (HDE), as our baselines, which are described in detail as follows.

\begin{itemize}
    \item \textbf{RED} randomly removes an edge in the innermost core of graph, and recalculates the innermost core; these two steps are performed iteratively until the original innermost core collapses. Since RED is a random algorithm, it is implemented 10 times in each experiment, and the mean values of the performance metrics are recorded.  
    
    \item \textbf{CKC} was originally proposed to find the minimal node set whose removal results in an empty \emph{k}-core. Here, we transform the goal of this problem as: finding the minimal node set whose removal results in as many nodes in the innermost core that change the $k$-shell values as possible. To make the algorithm more practical, this method first obtains the corona graph, and then finds the minimal node set whose removal results in the decreasing of $k$-shell values of all the nodes in the corona graph, and then recalculate the innermost core and the corona graph; these two steps are performed iteratively, until the original innermost core collapses.
    
    \item \textbf{HDN} removes the nodes with the highest degree in the innermost core subgraph, and then recalculates the innermost core; these two steps are performed iteratively, until the original innermost core collapses. 
    
     \item \textbf{HDE} is a variant of HDN, which removes the edges with the highest degree in the innermost core subgraph, where the degree of an edge is denoted as the sum of terminal-node degrees, and then recalculate the innermost core; these two steps are performed iteratively, until the original innermost core collapses.
\end{itemize}
    

\begin{table*}[t]
\caption{The attack results of CKC and HDN as the node-deleting strategies. We can see that both CKC and HDN have relatively high ECR and FAR, indicating the high cost and low precision of these two methods.}
\label{tab:compare3}
\begin{tabular*}{\hsize}{@{}@{\extracolsep{\fill}}ccccc|cccc@{}}
\bottomrule[0.5mm]
\multirow{2}{*}{Network} & \multicolumn{4}{c}{CKC} & \multicolumn{4}{c}{HDN}\\ \cline{2-9} 
                         & NDN  & NDE & ECR(\%) &FAR(\%) & NDN  & NDE & ECR(\%) &FAR(\%)  \\ \hline
Dolphin                  & 11  & 79    & 49.6855 & 72.3214 & 4   & 40    & 25.1572 & 15.3846 \\
Deezer                   & 3   & 122   & 0.1315  & 0.4349  & 3   & 130   & 0.1402  & 0.2694  \\
Enron-Email              & 170 & 48055 & 26.1409 & 30.6361 & 2   & 1840  & 1.0009  & 3.2402  \\
GitHub                   & 1   & 7085  & 2.4515  & 18.1379 & 1   & 7085  & 2.4515  & 17.3487 \\
Gowalla                  & 1   & 306   & 0.0322  & 0.2579  & 1   & 14730 & 1.55    & 7.1602  \\
Human                    & 1   & 37    & 1.6849  & 59.0999 & 1   & 98    & 4.4627  & 73.5849 \\
Web                      & 1   & 117   & 0.0067  & 0.0709  & 1   & 110   & 0.0063  & 0.0055  \\
DBLP                     & 399 & 18137 & 36.5518 & 39.377  & 4   & 1943  & 3.9158  & 8.3683  \\
FoodWeb                  & 1   & 49    & 2.3267  & 69.3949 & 1   & 110   & 5.2232  & 67.2727 \\
Proteins                 & 32  & 1001  & 16.2791 & 22.0864 & 3   & 187   & 3.0411  & 6.6088  \\
Autonomous               & 1   & 523   & 1.0798  & 2.2863  & 1   & 697   & 1.439   & 2.4681  \\
Brightkite               & 1   & 267   & 0.1247  & 0.529   & 1   & 386   & 0.1803  & 0.3668  \\
Cora                     & 1   & 60    & 0.0673  & 0.2374  & 1   & 52    & 0.0583  & 0.1253  \\
JUNG-Javax               & 1   & 148   & 0.2943  & 2.9081  & 1   & 204   & 0.4056  & 0.7519  \\
US-Power                 & 2   & 18    & 0.273   & 0.2631  & 2   & 16    & 0.2426  & 0.0406 \\ \hline
Average                  & 42  & 5067  & 9.142   & 21.202  & 1.8 & 1842  & 3.285   & 13.533     \\
\toprule[0.5mm]
\end{tabular*}
\end{table*}

\begin{table*}[t]
\caption{The attack results of RED, HDE, COREATTACK, and GreedyCOREATTACK as edge-deleting strategies. For all the four methods, we have NDN=0 and FAR=0, as well as much lower ECR, for all the datasets, indicating their natural superiority compared with node-deleting strategies.}
\label{tab:compare2}
\begin{tabular*}{\hsize}{@{}@{\extracolsep{\fill}}cc|cc|cc|cc|cc|c@{}}
\bottomrule[0.5mm]
\multirow{2}{*}{Network} &\multicolumn{1}{c}{}     & \multicolumn{2}{c}{RED} & \multicolumn{2}{c}{HDE} & \multicolumn{2}{c}{COREATTACK} & \multicolumn{2}{c}{GreedyCOREATTACK} &         \\ \cline{2-11}
                         & NDN & NDE      & ECR(\%)      & NDE       & ECR(\%)       & NDE          & ECR(\%)         & NDE             & ECR(\%)            & FAR(\%) \\ \hline
Dolphin                  & 0   & 19       & 11.761        & 17        & 10.6918       & \underline{\textbf{15}}           & \underline{\textbf{9.434}}           & \underline{\textbf{15}}              & \underline{\textbf{9.434}}              & 0       \\
Deezer                   & 0   & 17       & 0.0179         & 9         & 0.0097        & \underline{\textbf{4}}            & \underline{\textbf{0.0043}}          & \underline{\textbf{4}}               & \underline{\textbf{0.0043}}             & 0       \\
Enron-Email              & 0   & 197      & 0.1069          & 150       & 0.0816        & \underline{\textbf{32}}           & \underline{\textbf{0.0174}}          & 35              & 0.019              & 0       \\
GitHub                   & 0   & 173      & 0.0599         & 147       & 0.0509        & 48           & 0.0166          & \underline{\textbf{31}}              & \underline{\textbf{0.0107}}             & 0       \\
Gowalla                  & 0   & 57       & 0.006        & 40        & 0.0042        & \underline{\textbf{2}}            & \underline{\textbf{0.0002}}          & \underline{\textbf{2}}               & \underline{\textbf{0.0002}}             & 0       \\
Human                    & 0   & 6        & 0.2732         & 6         & 0.2732        & \underline{\textbf{1}}            & \underline{\textbf{0.0455}}          & \underline{\textbf{1}}               & \underline{\textbf{0.0455}}             & 0       \\
Web                      & 0   & \underline{\textbf{1}}        & \underline{\textbf{0.0001}}       & \underline{\textbf{1}}         & \underline{\textbf{0.0001}}        & \underline{\textbf{1}}            & \underline{\textbf{0.0001}}          & \underline{\textbf{1}}               & \underline{\textbf{0.0001}}             & 0       \\
DBLP                     & 0   & 441      & 0.8879         & 410       & 0.8263        & 240          & 0.4837          & \underline{\textbf{215}}             & \underline{\textbf{0.4333}}             & 0       \\
FoodWeb                  & 0   & 10       & 0.4748         & 6         & 0.2849        & \underline{\textbf{1}}            & \underline{\textbf{0.0475}}          & \underline{\textbf{1}}               & \underline{\textbf{0.0475}}             & 0       \\
Proteins                 & 0   & 33       & 0.5334         & 24        & 0.3903        & 20           & 0.3253          & 15              & 0.2439             & 0       \\
Autonomous               & 0   & 62       & 0.128          & 39        & 0.0805        & 18           & 0.0372          & \underline{\textbf{13}}              & \underline{\textbf{0.0268}}             & 0       \\
Brightkite               & 0   & 88       & 0.0411         & 18        & 0.0084         & 7            & 0.0033          & \underline{\textbf{3}}               & \underline{\textbf{0.0014}}             & 0       \\
Cora                     & 0   & 8        & 0.008        & 3         & 0.0034        & \underline{\textbf{2}}            & \underline{\textbf{0.0022}}          & \underline{\textbf{2}}               & \underline{\textbf{0.0022}}             & 0       \\
JUNG-Javax               & 0   & 23       & 0.0447         & 23        & 0.0457        & 26           & 0.0517          & \underline{\textbf{8}}               & \underline{\textbf{0.0159}}             & 0       \\
US-Power                 & 0   & 3        & 0.044         & 4         & 0.0607        & \underline{\textbf{2}}            & \underline{\textbf{0.0303}}          & \underline{\textbf{2}}               & \underline{\textbf{0.0303}}             & 0      \\ \hline
Average       & 0 & 67.5   & 0.954  & 92.4 & 1.5 & 27.9 & 0.7 & \underline{\textbf{23.2}} & \underline{\textbf{0.688}} &0\\
\toprule[0.5mm]
\end{tabular*}
\end{table*}

\begin{figure*}[t]
    \centering
    \subfigure{\includegraphics[width=0.24\linewidth]{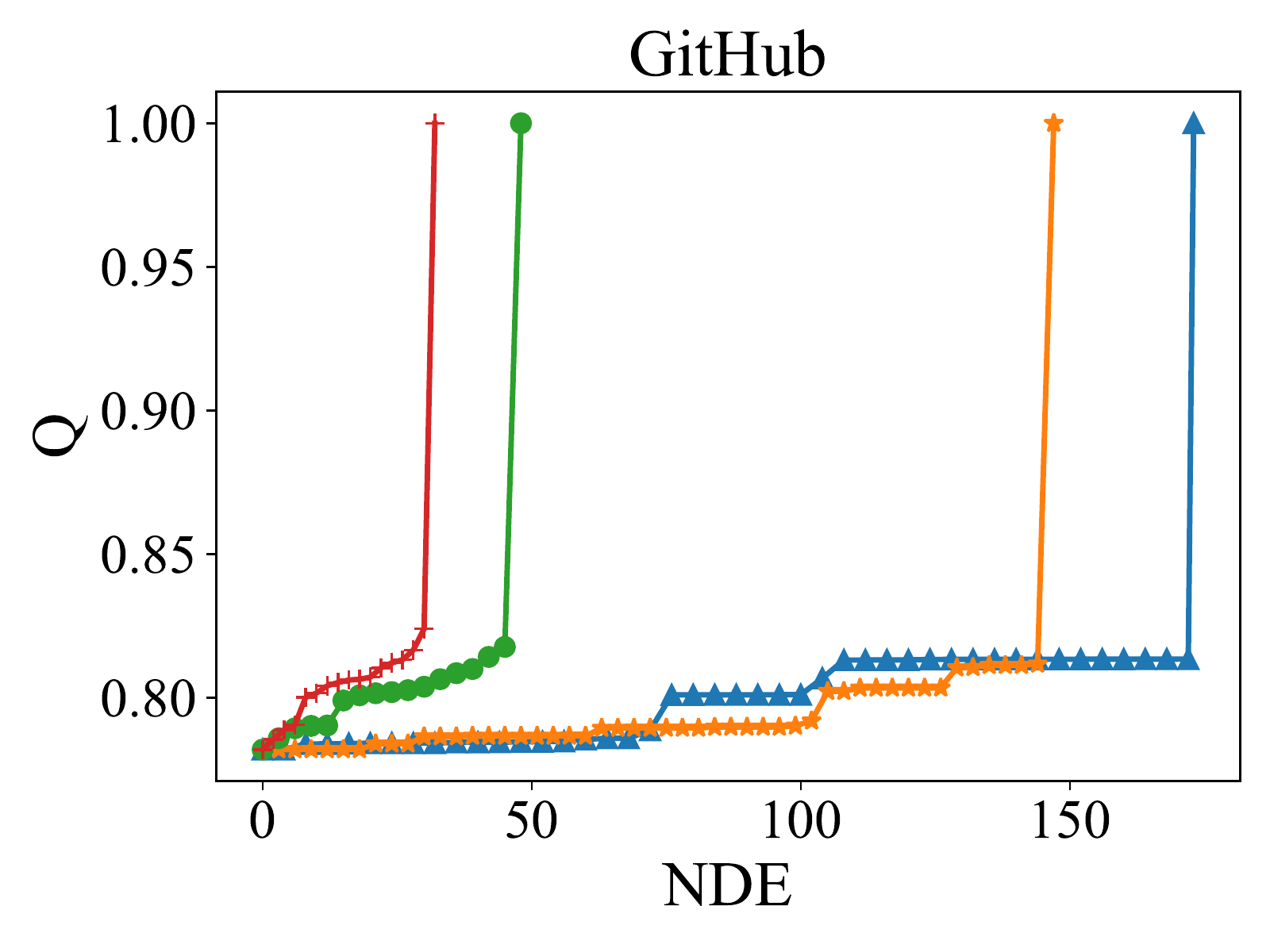}}
    \subfigure{\includegraphics[width=0.24\linewidth]{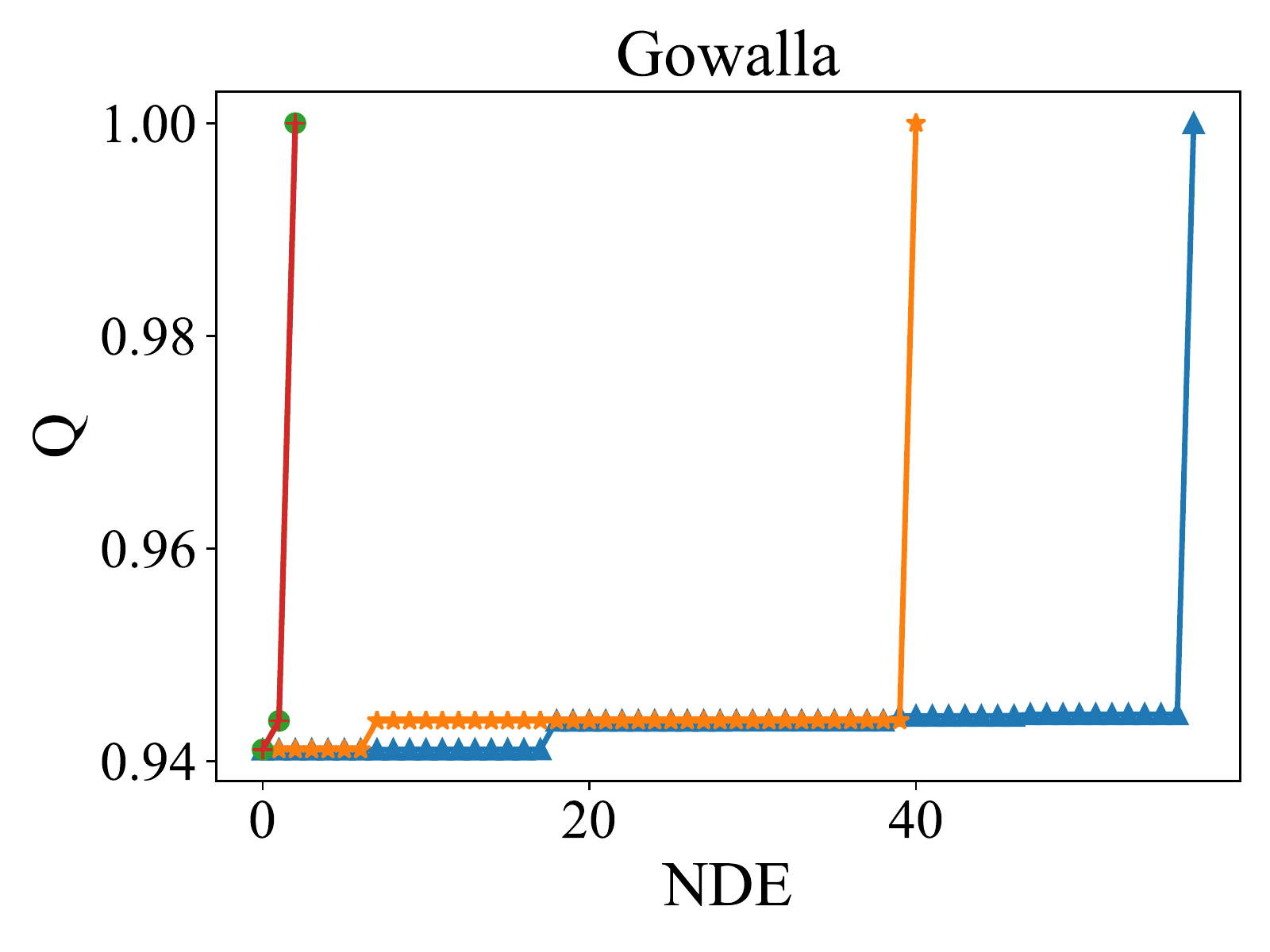}}
    \subfigure{\includegraphics[width=0.24\linewidth]{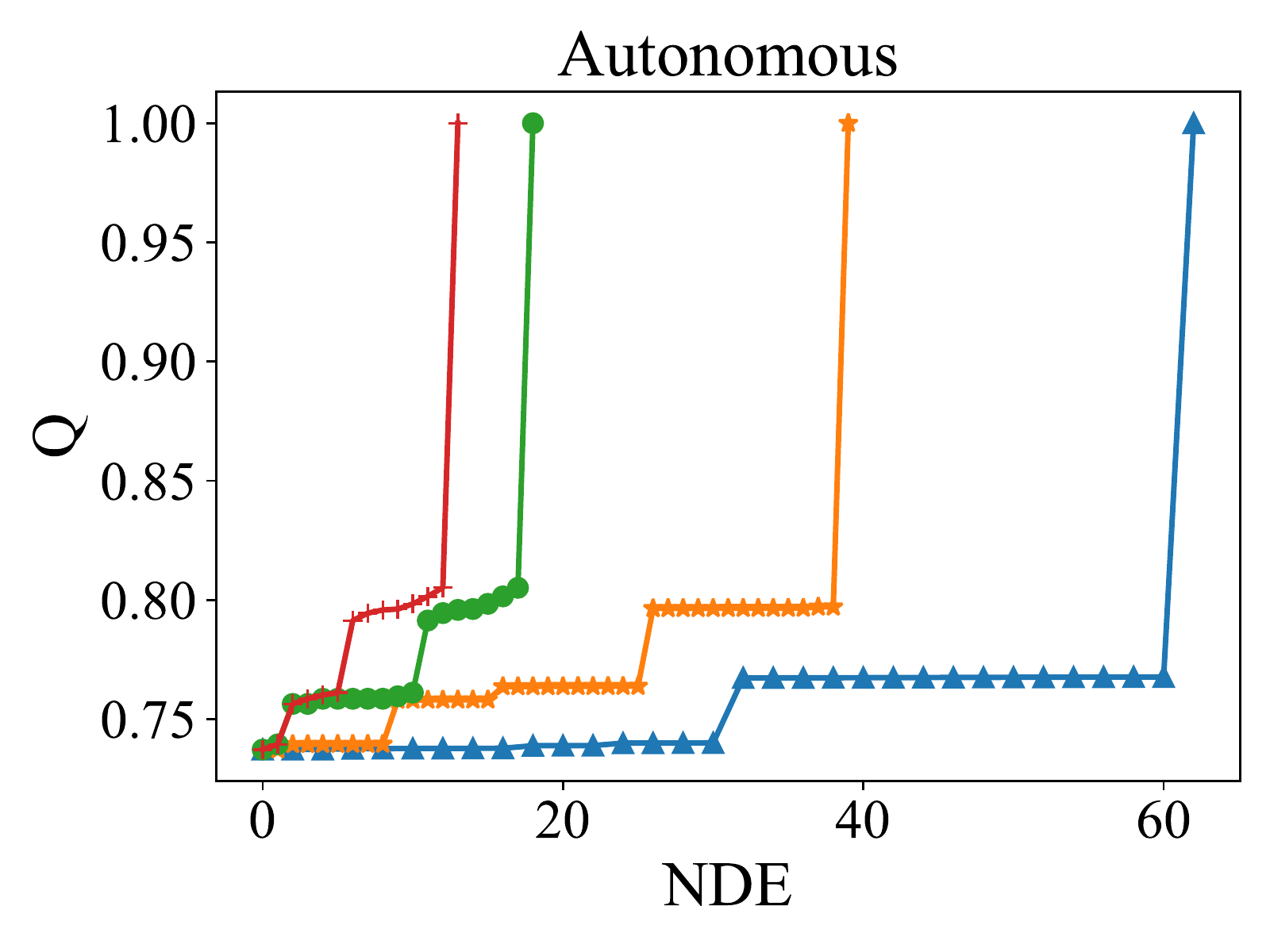}}
    \subfigure{\includegraphics[width=0.24\linewidth]{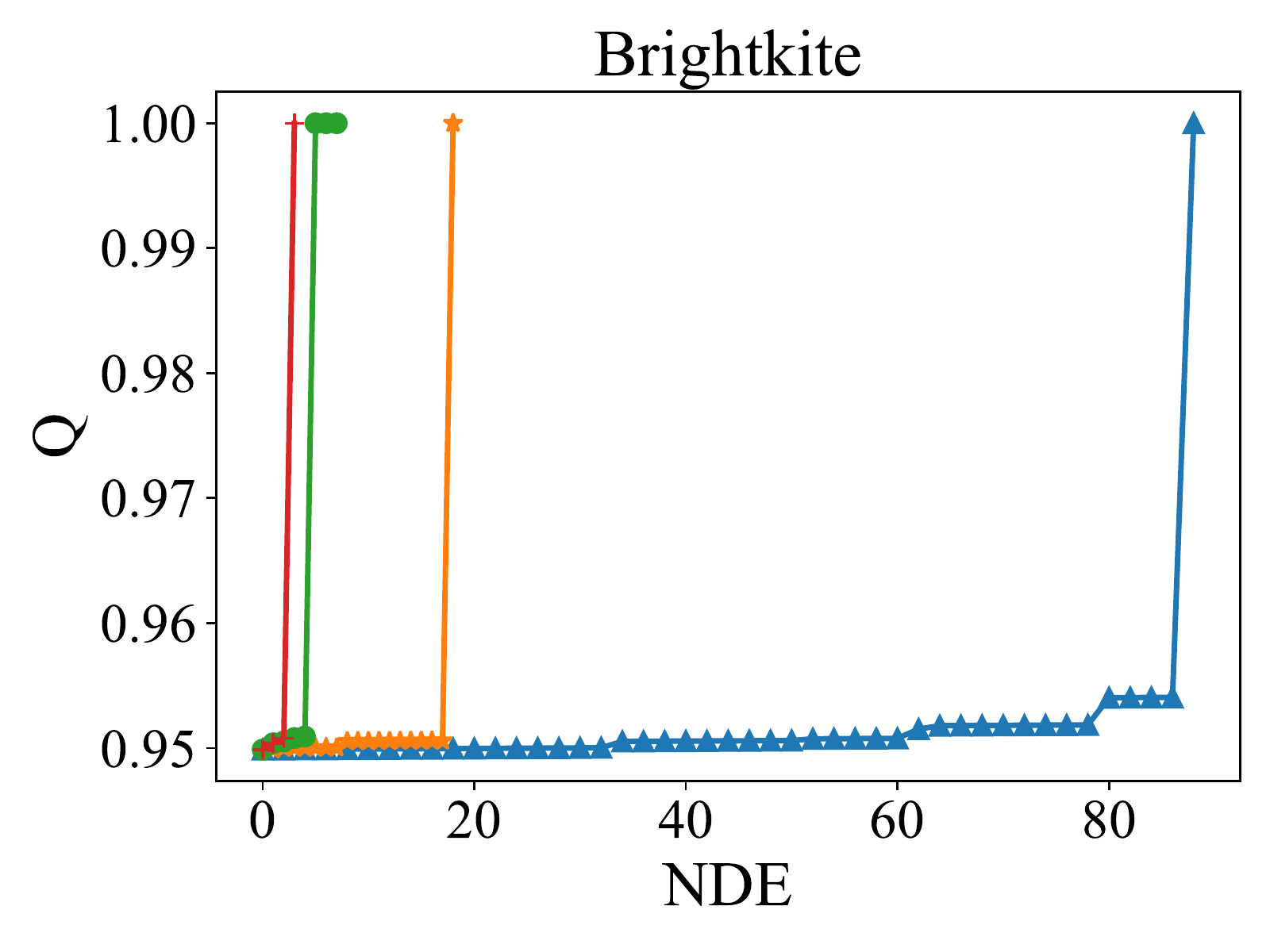}}
    \subfigure{\includegraphics[width=0.5\linewidth]{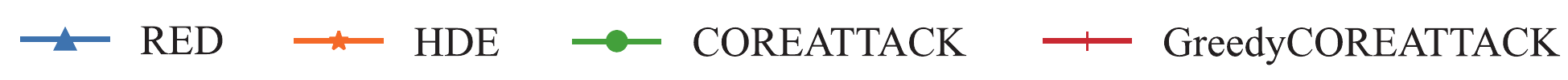}}
    \vspace{-4mm}
    \caption{The relationship between $Q$ index and NDE on four networks including GitHub, Gowalla, Autonomous, and Brightkite, by adopting RED, HDE, COREATTACK and GreedyCOREATTACK.}
    \label{figure:Q}
\end{figure*}

\subsection{Results and Discussion}
The experimental results are shown in Table~\ref{tab:compare3} and Table~\ref{tab:compare2}, which contrastively exhibit the performance comparison of COREATTACK and GreedyCOREATTACK with the four baselines on the 15 real-world networks.

It is obvious that all the attack methods based on edge deletion, i.e., RED, HDE, COREATTACK, and GreedyCOREATTACK, behave much better than those existent attack methods based on node deletion, i.e., CKC and HDN, in terms of significantly lower NDN, NDE, ECR, and FAR. Note that for all the four edge-deleting strategies, both NDN and FAR even equal to 0, indicating that they will not make any node isolated from the network and will also not decrease the $k$-shell value of any node outside the innermost core.


In particular, compared with CKC and HD, COREATTACK and GreedyCOREATTACK are able to destroy the innermost core with much fewer edges deleted, with a relatively decrease of two orders of magnitude on average (from 5067 of CKC decrease to 23 of GreedyCOREATTACK on 15 networks).
One can easily perceive that, since the innermost core of a large-scale graph is a subgraph that could contain thousands or even tens of thousands of edges, as shown in Table~\ref{tab:dataset}, deleting a certain node from this subgraph will result in the disappearance of its contiguous edges that could account for a large percentage of the edges of the entire network.
Take the Enron-Email network in Table~\ref{tab:compare3} for example, 
in order to destroy the innermost core, CKC needs to delete 48055 edges (about 26\% of all the edge) and HDN needs to delete 1840 edges (about 1\%). In fact, the innermost core of this network is 43-core, which means each node in this innermost core owns at least 43 edges without considering the edges that have been stripped during \emph{k}-core decomposition. If a node in this innermost core is deleted, it is not difficult to imagine that a large number of adjacent edges will be removed together, which will cause an obvious damage to the network structure. By comparison, as shown in Table~\ref{tab:compare2}, our COREATTACK only deletes 32 edges (about 0.0174\%). In other words, our attack is much easier to be realized, and meanwhile the resulting perturbation on the network structure is negligible and thus has a low risk of being detected.

Moreover, as we can see from Table~\ref{tab:compare3} and Table~\ref{tab:compare2}, different from the attack strategies based on edge deletion, both CKC and HDN always have $\text{FAR}>0$ which is even larger than 50\% on certain networks, such as Human and FoodWeb. This phenomenon suggests that deleting nodes not only make the innermost core collapse, but also influences the $k$-shell values of the nodes outside the innermost core. This is reasonable, since a node inside the innermost core could be also connected to a large number of nodes outside the innermost core. Taking CKC on GitHub as an example, we only need to remove one node to make the innermost core collapse. Since this node has 7085 neighbors, we also delete 7085 links in this network. Note that the innermost core only contains 360 nodes, as shown in Table~\ref{tab:dataset}, the deletion of 7085 links will definitely influence a large number of nodes outside the innermost core. In other words, the attack strategies based on node deletion are less efficient and precise, which is the exact reason why we would like to develop edge-deleting attack methods in this work.

Now, let's focus on the results obtained by the four attack methods based on edge deletion, as shown in Table~\ref{tab:compare2}. Still, our COREATTACK and GreedyCOREATTACK behaves significantly better than RED and HDE, in terms of much smaller NDE and ECR; while GreedyCOREATTACK is slightly better than COREATTACK. Such results indicate that the edges in corona graph indeed play key roles in maintaining the innermost core of the network, removal of which will make the core collapse quickly. It's quite interesting to see that we even need to only remove 1 edge to disintegrate the innermost cores of Human, Web and FoodWeb networks, with their size equal to 60, 102, and 73, respectively. For the Web network, we have $k_{\max}=101$ and $\left|V_I\right|=102$, as shown in Table~\ref{tab:dataset}, which indicates that the innermost core of this network is a fully connected graph. Therefore, removal of any link in this innermost core will make it collapse, no matter which strategy is used. For the other two networks, both RED and HDE need to remove more than 1 edges to make the innermost core collapse, indicating the advantages of our methods.


To visualize the process of the four edge-deletion attack methods, i.e., RED, HDE, COREATTACK, and GreedyCOREATTACK, we give the relationship between $Q$ index and NDE on four networks including GitHub, Gowalla, Autonomous, and Brightkite, in Figure~\ref{figure:Q}. Similar results can be obtained on the other networks. Based on the results presented in Table~\ref{tab:compare2}, we can see that: first, the collapse of innermost core is basically completed after the deletion of a limited number of edges; second, previous deletion has weak influence on the collapse of the innermost core while \emph{the last straw breaks the camel's back} and can make the \emph{k}-core collapse at once. The fact is that the previous deletion destroys the structure of the innermost core to certain extent but does not fundamentally influence the innermost core, and the deletion of the last edge achieves the critical point of NDE to finish the last strike on the collapse of the innermost core. By comparison, our COREATTACK and GreedyCOREATTACK strategies have the much smaller critical points, i.e., significantly fewer edges are needed to be removed to break up the innermost core of a graph by using our proposed attack methods.   




\section{conclusion}\label{sec:conclusion}
In this paper, we verify that the \emph{k}-core structure of many real-world networks is quite vulnerable under adversarial attacks. Based on the theoretical analysis of \emph{k}-core structure, we propose a heuristic strategy named COREATTACK to disintegrate the innermost core (the $k$-core with the largest $k$ value) of 15 real-world networks by deleting a small number of edges. We demonstrate that COREATTACK is more effective than the existent node-deletion attack strategies, and is also better than the basic 
edge-deletion attack strategies, in terms of much smaller Edge Change Rate (ECR) and False Attack Rate (FAR). Moreover, in order to further improve the effectiveness, we design a greedy version of COREATTACK, namely GreedyCOREATTACK, which is even better than COREATTACK, achieving the state-of-the-art attack results. 

Since \emph{k}-core structure plays significant roles in information diffusion, graph algorithm design etc., it would be of great interest to propose strategies to defense against such adversarial attacks, which belongs to our future work.

\clearpage



\begin{thebibliography}{38}


\ifx \showCODEN    \undefined \def \showCODEN     #1{\unskip}     \fi
\ifx \showDOI      \undefined \def \showDOI       #1{#1}\fi
\ifx \showISBNx    \undefined \def \showISBNx     #1{\unskip}     \fi
\ifx \showISBNxiii \undefined \def \showISBNxiii  #1{\unskip}     \fi
\ifx \showISSN     \undefined \def \showISSN      #1{\unskip}     \fi
\ifx \showLCCN     \undefined \def \showLCCN      #1{\unskip}     \fi
\ifx \shownote     \undefined \def \shownote      #1{#1}          \fi
\ifx \showarticletitle \undefined \def \showarticletitle #1{#1}   \fi
\ifx \showURL      \undefined \def \showURL       {\relax}        \fi
\providecommand\bibfield[2]{#2}
\providecommand\bibinfo[2]{#2}
\providecommand\natexlab[1]{#1}
\providecommand\showeprint[2][]{arXiv:#2}

\bibitem[\protect\citeauthoryear{Altaf-Ul-Amine, Nishikata, Korna, Miyasato,
  Shinbo, Arifuzzaman, Wada, Maeda, Oshima, Mori, et~al\mbox{.}}{Altaf-Ul-Amine
  et~al\mbox{.}}{2003}]%
        {altaf2003prediction}
\bibfield{author}{\bibinfo{person}{Md Altaf-Ul-Amine}, \bibinfo{person}{Kensaku
  Nishikata}, \bibinfo{person}{Toshihiro Korna}, \bibinfo{person}{Teppei
  Miyasato}, \bibinfo{person}{Yoko Shinbo}, \bibinfo{person}{Md Arifuzzaman},
  \bibinfo{person}{Chieko Wada}, \bibinfo{person}{Maki Maeda},
  \bibinfo{person}{Taku Oshima}, \bibinfo{person}{Hirotada Mori},
  {et~al\mbox{.}}} \bibinfo{year}{2003}\natexlab{}.
\newblock \showarticletitle{Prediction of protein functions based on k-cores of
  protein-protein interaction networks and amino acid sequences}.
\newblock \bibinfo{journal}{\emph{Genome Informatics}}  \bibinfo{volume}{14}
  (\bibinfo{year}{2003}), \bibinfo{pages}{498--499}.
\newblock


\bibitem[\protect\citeauthoryear{Batagelj and Zaversnik}{Batagelj and
  Zaversnik}{2003}]%
        {batagelj2003m}
\bibfield{author}{\bibinfo{person}{Vladimir Batagelj} {and}
  \bibinfo{person}{Matjaz Zaversnik}.} \bibinfo{year}{2003}\natexlab{}.
\newblock \bibinfo{title}{An O (m) algorithm for cores decomposition of
  networks}.
\newblock
\newblock


\bibitem[\protect\citeauthoryear{Baxter, Dorogovtsev, Lee, Mendes, and
  Goltsev}{Baxter et~al\mbox{.}}{2015}]%
        {baxter2015critical}
\bibfield{author}{\bibinfo{person}{GJ Baxter}, \bibinfo{person}{SN
  Dorogovtsev}, \bibinfo{person}{K-E Lee}, \bibinfo{person}{JFF Mendes}, {and}
  \bibinfo{person}{AV Goltsev}.} \bibinfo{year}{2015}\natexlab{}.
\newblock \showarticletitle{Critical dynamics of the k-core pruning process}.
\newblock \bibinfo{journal}{\emph{Physical Review X}} \bibinfo{volume}{5},
  \bibinfo{number}{3} (\bibinfo{year}{2015}), \bibinfo{pages}{031017}.
\newblock


\bibitem[\protect\citeauthoryear{Bhagat, Cormode, and Muthukrishnan}{Bhagat
  et~al\mbox{.}}{2011}]%
        {bhagat2011node}
\bibfield{author}{\bibinfo{person}{Smriti Bhagat}, \bibinfo{person}{Graham
  Cormode}, {and} \bibinfo{person}{S Muthukrishnan}.}
  \bibinfo{year}{2011}\natexlab{}.
\newblock \showarticletitle{Node classification in social networks}.
\newblock In \bibinfo{booktitle}{\emph{Social network data analytics}}.
  \bibinfo{publisher}{Springer}, \bibinfo{address}{Boston},
  \bibinfo{pages}{115--148}.
\newblock


\bibitem[\protect\citeauthoryear{Bogu\~n\'a, Pastor-Satorras,
  D\'{\i}az-Guilera, and Arenas}{Bogu\~n\'a et~al\mbox{.}}{2004}]%
        {PhysRevE.70.056122}
\bibfield{author}{\bibinfo{person}{Mari\'an Bogu\~n\'a},
  \bibinfo{person}{Romualdo Pastor-Satorras}, \bibinfo{person}{Albert
  D\'{\i}az-Guilera}, {and} \bibinfo{person}{Alex Arenas}.}
  \bibinfo{year}{2004}\natexlab{}.
\newblock \showarticletitle{Models of social networks based on social distance
  attachment}.
\newblock \bibinfo{journal}{\emph{Phys. Rev. E}}  \bibinfo{volume}{70}
  (\bibinfo{date}{Nov} \bibinfo{year}{2004}), \bibinfo{pages}{056122}.
\newblock
Issue 5.
\urldef\tempurl%
\url{https://doi.org/10.1103/PhysRevE.70.056122}
\showDOI{\tempurl}


\bibitem[\protect\citeauthoryear{Borge-Holthoefer and Moreno}{Borge-Holthoefer
  and Moreno}{2012}]%
        {PhysRevE.85.026116}
\bibfield{author}{\bibinfo{person}{Javier Borge-Holthoefer} {and}
  \bibinfo{person}{Yamir Moreno}.} \bibinfo{year}{2012}\natexlab{}.
\newblock \showarticletitle{Absence of influential spreaders in rumor
  dynamics}.
\newblock \bibinfo{journal}{\emph{Phys. Rev. E}}  \bibinfo{volume}{85}
  (\bibinfo{date}{Feb} \bibinfo{year}{2012}), \bibinfo{pages}{026116}.
\newblock
Issue 2.
\urldef\tempurl%
\url{https://doi.org/10.1103/PhysRevE.85.026116}
\showDOI{\tempurl}


\bibitem[\protect\citeauthoryear{Cha, Haddadi, Benevenuto, and Gummadi}{Cha
  et~al\mbox{.}}{2010}]%
        {Cha_Haddadi_Benevenuto_Gummadi_2010}
\bibfield{author}{\bibinfo{person}{Meeyoung Cha}, \bibinfo{person}{Hamed
  Haddadi}, \bibinfo{person}{Fabricio Benevenuto}, {and}
  \bibinfo{person}{Krishna Gummadi}.} \bibinfo{year}{2010}\natexlab{}.
\newblock \showarticletitle{Measuring User Influence in Twitter: The Million
  Follower Fallacy}.
\newblock \bibinfo{journal}{\emph{Proceedings of the International AAAI
  Conference on Web and Social Media}} \bibinfo{volume}{4}, \bibinfo{number}{1}
  (\bibinfo{date}{May} \bibinfo{year}{2010}), \bibinfo{pages}{10--17}.
\newblock
\urldef\tempurl%
\url{https://ojs.aaai.org/index.php/ICWSM/article/view/14033}
\showURL{%
\tempurl}


\bibitem[\protect\citeauthoryear{Chen, Zhang, Chen, Du, and Xuan}{Chen
  et~al\mbox{.}}{2021}]%
        {9531428}
\bibfield{author}{\bibinfo{person}{Jinyin Chen}, \bibinfo{person}{Jian Zhang},
  \bibinfo{person}{Zhi Chen}, \bibinfo{person}{Min Du}, {and}
  \bibinfo{person}{Qi Xuan}.} \bibinfo{year}{2021}\natexlab{}.
\newblock \showarticletitle{Time-aware Gradient Attack on Dynamic Network Link
  Prediction}.
\newblock \bibinfo{journal}{\emph{IEEE Transactions on Knowledge and Data
  Engineering}} \bibinfo{volume}{0}, \bibinfo{number}{0}
  (\bibinfo{year}{2021}), \bibinfo{pages}{1--1}.
\newblock
\urldef\tempurl%
\url{https://doi.org/10.1109/TKDE.2021.3110580}
\showDOI{\tempurl}


\bibitem[\protect\citeauthoryear{Cheng, Ke, Chu, and Özsu}{Cheng
  et~al\mbox{.}}{2011}]%
        {5767911}
\bibfield{author}{\bibinfo{person}{James Cheng}, \bibinfo{person}{Yiping Ke},
  \bibinfo{person}{Shumo Chu}, {and} \bibinfo{person}{M.~Tamer Özsu}.}
  \bibinfo{year}{2011}\natexlab{}.
\newblock \showarticletitle{Efficient core decomposition in massive networks}.
  In \bibinfo{booktitle}{\emph{2011 IEEE 27th International Conference on Data
  Engineering}}. \bibinfo{publisher}{IEEE}, \bibinfo{address}{ch},
  \bibinfo{pages}{51--62}.
\newblock
\urldef\tempurl%
\url{https://doi.org/10.1109/ICDE.2011.5767911}
\showDOI{\tempurl}


\bibitem[\protect\citeauthoryear{Cormen, Leiserson, Rivest, and Stein}{Cormen
  et~al\mbox{.}}{2009}]%
        {cormen2009introduction}
\bibfield{author}{\bibinfo{person}{Thomas~H Cormen}, \bibinfo{person}{CE
  Leiserson}, \bibinfo{person}{RL Rivest}, {and} \bibinfo{person}{C Stein}.}
  \bibinfo{year}{2009}\natexlab{}.
\newblock \bibinfo{title}{Introduction to algorithms,‖ Cambridge, Mass}.
\newblock
\newblock


\bibitem[\protect\citeauthoryear{Dorogovtsev, Goltsev, and Mendes}{Dorogovtsev
  et~al\mbox{.}}{2006}]%
        {dorogovtsev2006k}
\bibfield{author}{\bibinfo{person}{SN Dorogovtsev}, \bibinfo{person}{AV
  Goltsev}, {and} \bibinfo{person}{JFF Mendes}.}
  \bibinfo{year}{2006}\natexlab{}.
\newblock \showarticletitle{k-Core architecture and k-core percolation on
  complex networks}.
\newblock \bibinfo{journal}{\emph{Physica D: Nonlinear Phenomena}}
  \bibinfo{volume}{224}, \bibinfo{number}{1-2} (\bibinfo{year}{2006}),
  \bibinfo{pages}{7--19}.
\newblock


\bibitem[\protect\citeauthoryear{Erd{\H{o}}s and Hajnal}{Erd{\H{o}}s and
  Hajnal}{1966}]%
        {erdHos1966chromatic}
\bibfield{author}{\bibinfo{person}{Paul Erd{\H{o}}s} {and}
  \bibinfo{person}{Andr{\'a}s Hajnal}.} \bibinfo{year}{1966}\natexlab{}.
\newblock \showarticletitle{On chromatic number of graphs and set-systems}.
\newblock \bibinfo{journal}{\emph{Acta Mathematica Academiae Scientiarum
  Hungarica}} \bibinfo{volume}{17}, \bibinfo{number}{1-2}
  (\bibinfo{year}{1966}), \bibinfo{pages}{61--99}.
\newblock


\bibitem[\protect\citeauthoryear{Feige}{Feige}{1998}]%
        {feige1998threshold}
\bibfield{author}{\bibinfo{person}{Uriel Feige}.}
  \bibinfo{year}{1998}\natexlab{}.
\newblock \showarticletitle{A threshold of ln n for approximating set cover}.
\newblock \bibinfo{journal}{\emph{Journal of the ACM (JACM)}}
  \bibinfo{volume}{45}, \bibinfo{number}{4} (\bibinfo{year}{1998}),
  \bibinfo{pages}{634--652}.
\newblock


\bibitem[\protect\citeauthoryear{Filho, Machicao, and Bruno}{Filho
  et~al\mbox{.}}{2018}]%
        {filho2018hierarchical}
\bibfield{author}{\bibinfo{person}{Humberto~A Filho}, \bibinfo{person}{Jeaneth
  Machicao}, {and} \bibinfo{person}{Odemir~M Bruno}.}
  \bibinfo{year}{2018}\natexlab{}.
\newblock \showarticletitle{A hierarchical model of metabolic machinery based
  on the k core decomposition of plant metabolic networks}.
\newblock \bibinfo{journal}{\emph{PloS one}} \bibinfo{volume}{13},
  \bibinfo{number}{5} (\bibinfo{year}{2018}), \bibinfo{pages}{e0195843}.
\newblock


\bibitem[\protect\citeauthoryear{Fortunato}{Fortunato}{2010}]%
        {fortunato2010community}
\bibfield{author}{\bibinfo{person}{Santo Fortunato}.}
  \bibinfo{year}{2010}\natexlab{}.
\newblock \showarticletitle{Community detection in graphs}.
\newblock \bibinfo{journal}{\emph{Physics reports}} \bibinfo{volume}{486},
  \bibinfo{number}{3-5} (\bibinfo{year}{2010}), \bibinfo{pages}{75--174}.
\newblock


\bibitem[\protect\citeauthoryear{Garc{\'\i}a-Algarra, Pastor, Iriondo, and
  Galeano}{Garc{\'\i}a-Algarra et~al\mbox{.}}{2017}]%
        {garcia2017ranking}
\bibfield{author}{\bibinfo{person}{Javier Garc{\'\i}a-Algarra},
  \bibinfo{person}{Juan~Manuel Pastor}, \bibinfo{person}{Jos{\'e}~Mar{\'\i}a
  Iriondo}, {and} \bibinfo{person}{Javier Galeano}.}
  \bibinfo{year}{2017}\natexlab{}.
\newblock \showarticletitle{Ranking of critical species to preserve the
  functionality of mutualistic networks using the k-core decomposition}.
\newblock \bibinfo{journal}{\emph{PeerJ}}  \bibinfo{volume}{5}
  (\bibinfo{year}{2017}), \bibinfo{pages}{e3321}.
\newblock


\bibitem[\protect\citeauthoryear{Giatsidis, Thilikos, and
  Vazirgiannis}{Giatsidis et~al\mbox{.}}{2011}]%
        {giatsidis2011evaluating}
\bibfield{author}{\bibinfo{person}{Christos Giatsidis},
  \bibinfo{person}{Dimitrios~M Thilikos}, {and} \bibinfo{person}{Michalis
  Vazirgiannis}.} \bibinfo{year}{2011}\natexlab{}.
\newblock \showarticletitle{Evaluating cooperation in communities with the
  k-core structure}. In \bibinfo{booktitle}{\emph{2011 International conference
  on advances in social networks analysis and mining}}. IEEE,
  \bibinfo{publisher}{IEEE}, \bibinfo{address}{chi}, \bibinfo{pages}{87--93}.
\newblock


\bibitem[\protect\citeauthoryear{Giatsidis, Thilikos, and
  Vazirgiannis}{Giatsidis et~al\mbox{.}}{2013}]%
        {giatsidis2013d}
\bibfield{author}{\bibinfo{person}{Christos Giatsidis},
  \bibinfo{person}{Dimitrios~M Thilikos}, {and} \bibinfo{person}{Michalis
  Vazirgiannis}.} \bibinfo{year}{2013}\natexlab{}.
\newblock \showarticletitle{D-cores: measuring collaboration of directed graphs
  based on degeneracy}.
\newblock \bibinfo{journal}{\emph{Knowledge and information systems}}
  \bibinfo{volume}{35}, \bibinfo{number}{2} (\bibinfo{year}{2013}),
  \bibinfo{pages}{311--343}.
\newblock


\bibitem[\protect\citeauthoryear{Goltsev, Dorogovtsev, and Mendes}{Goltsev
  et~al\mbox{.}}{2006}]%
        {goltsev2006k}
\bibfield{author}{\bibinfo{person}{Alexander~V Goltsev},
  \bibinfo{person}{Sergey~N Dorogovtsev}, {and} \bibinfo{person}{Jose
  Ferreira~F Mendes}.} \bibinfo{year}{2006}\natexlab{}.
\newblock \showarticletitle{k-core (bootstrap) percolation on complex networks:
  Critical phenomena and nonlocal effects}.
\newblock \bibinfo{journal}{\emph{Physical Review E}} \bibinfo{volume}{73},
  \bibinfo{number}{5} (\bibinfo{year}{2006}), \bibinfo{pages}{056101}.
\newblock


\bibitem[\protect\citeauthoryear{Guimera, Danon, Diaz-Guilera, Giralt, and
  Arenas}{Guimera et~al\mbox{.}}{2003}]%
        {guimera2003self}
\bibfield{author}{\bibinfo{person}{Roger Guimera}, \bibinfo{person}{Leon
  Danon}, \bibinfo{person}{Albert Diaz-Guilera}, \bibinfo{person}{Francesc
  Giralt}, {and} \bibinfo{person}{Alex Arenas}.}
  \bibinfo{year}{2003}\natexlab{}.
\newblock \showarticletitle{Self-similar community structure in a network of
  human interactions}.
\newblock \bibinfo{journal}{\emph{Physical review E}} \bibinfo{volume}{68},
  \bibinfo{number}{6} (\bibinfo{year}{2003}), \bibinfo{pages}{065103}.
\newblock


\bibitem[\protect\citeauthoryear{Khaouid, Barsky, Srinivasan, and
  Thomo}{Khaouid et~al\mbox{.}}{2015}]%
        {khaouid2015k}
\bibfield{author}{\bibinfo{person}{Wissam Khaouid}, \bibinfo{person}{Marina
  Barsky}, \bibinfo{person}{Venkatesh Srinivasan}, {and} \bibinfo{person}{Alex
  Thomo}.} \bibinfo{year}{2015}\natexlab{}.
\newblock \showarticletitle{K-core decomposition of large networks on a single
  pc}.
\newblock \bibinfo{journal}{\emph{Proceedings of the VLDB Endowment}}
  \bibinfo{volume}{9}, \bibinfo{number}{1} (\bibinfo{year}{2015}),
  \bibinfo{pages}{13--23}.
\newblock


\bibitem[\protect\citeauthoryear{Kunegis}{Kunegis}{2013}]%
        {kunegis2013konect}
\bibfield{author}{\bibinfo{person}{J{\'e}r{\^o}me Kunegis}.}
  \bibinfo{year}{2013}\natexlab{}.
\newblock \showarticletitle{Konect: the koblenz network collection}. In
  \bibinfo{booktitle}{\emph{Proceedings of the 22nd international conference on
  World Wide Web}}. \bibinfo{publisher}{ACM}, \bibinfo{address}{chi},
  \bibinfo{pages}{1343--1350}.
\newblock


\bibitem[\protect\citeauthoryear{Leskovec and Krevl}{Leskovec and
  Krevl}{2014}]%
        {snapnets}
\bibfield{author}{\bibinfo{person}{Jure Leskovec} {and} \bibinfo{person}{Andrej
  Krevl}.} \bibinfo{year}{2014}\natexlab{}.
\newblock \bibinfo{title}{{SNAP Datasets}: {Stanford} Large Network Dataset
  Collection}.
\newblock \bibinfo{howpublished}{\url{http://snap.stanford.edu/data}}.
\newblock


\bibitem[\protect\citeauthoryear{Lusseau, Schneider, Boisseau, Haase, Slooten,
  and Dawson}{Lusseau et~al\mbox{.}}{2003}]%
        {lusseau2003bottlenose}
\bibfield{author}{\bibinfo{person}{David Lusseau}, \bibinfo{person}{Karsten
  Schneider}, \bibinfo{person}{Oliver~J Boisseau}, \bibinfo{person}{Patti
  Haase}, \bibinfo{person}{Elisabeth Slooten}, {and} \bibinfo{person}{Steve~M
  Dawson}.} \bibinfo{year}{2003}\natexlab{}.
\newblock \showarticletitle{The bottlenose dolphin community of Doubtful Sound
  features a large proportion of long-lasting associations}.
\newblock \bibinfo{journal}{\emph{Behavioral Ecology and Sociobiology}}
  \bibinfo{volume}{54}, \bibinfo{number}{4} (\bibinfo{year}{2003}),
  \bibinfo{pages}{396--405}.
\newblock


\bibitem[\protect\citeauthoryear{Malliaros, Giatsidis, Papadopoulos, and
  Vazirgiannis}{Malliaros et~al\mbox{.}}{2020}]%
        {malliaros2020core}
\bibfield{author}{\bibinfo{person}{Fragkiskos~D Malliaros},
  \bibinfo{person}{Christos Giatsidis}, \bibinfo{person}{Apostolos~N
  Papadopoulos}, {and} \bibinfo{person}{Michalis Vazirgiannis}.}
  \bibinfo{year}{2020}\natexlab{}.
\newblock \showarticletitle{The core decomposition of networks: Theory,
  algorithms and applications}.
\newblock \bibinfo{journal}{\emph{The VLDB Journal}} \bibinfo{volume}{29},
  \bibinfo{number}{1} (\bibinfo{year}{2020}), \bibinfo{pages}{61--92}.
\newblock


\bibitem[\protect\citeauthoryear{Matula and Beck}{Matula and Beck}{1983}]%
        {matula1983smallest}
\bibfield{author}{\bibinfo{person}{David~W Matula} {and}
  \bibinfo{person}{Leland~L Beck}.} \bibinfo{year}{1983}\natexlab{}.
\newblock \showarticletitle{Smallest-last ordering and clustering and graph
  coloring algorithms}.
\newblock \bibinfo{journal}{\emph{Journal of the ACM (JACM)}}
  \bibinfo{volume}{30}, \bibinfo{number}{3} (\bibinfo{year}{1983}),
  \bibinfo{pages}{417--427}.
\newblock


\bibitem[\protect\citeauthoryear{Miorandi and De~Pellegrini}{Miorandi and
  De~Pellegrini}{2010}]%
        {miorandi2010k}
\bibfield{author}{\bibinfo{person}{Daniele Miorandi} {and}
  \bibinfo{person}{Francesco De~Pellegrini}.} \bibinfo{year}{2010}\natexlab{}.
\newblock \showarticletitle{K-shell decomposition for dynamic complex
  networks}. In \bibinfo{booktitle}{\emph{8th International Symposium on
  Modeling and Optimization in Mobile, Ad Hoc, and Wireless Networks}}. IEEE,
  \bibinfo{publisher}{IEEE}, \bibinfo{address}{chi}, \bibinfo{pages}{488--496}.
\newblock


\bibitem[\protect\citeauthoryear{Rossi and Ahmed}{Rossi and Ahmed}{2015}]%
        {nr}
\bibfield{author}{\bibinfo{person}{Ryan~A. Rossi} {and}
  \bibinfo{person}{Nesreen~K. Ahmed}.} \bibinfo{year}{2015}\natexlab{}.
\newblock \showarticletitle{The Network Data Repository with Interactive Graph
  Analytics and Visualization}. In \bibinfo{booktitle}{\emph{AAAI}}.
  \bibinfo{publisher}{aaai}, \bibinfo{address}{chi}, \bibinfo{pages}{396--405}.
\newblock
\urldef\tempurl%
\url{https://networkrepository.com}
\showURL{%
\tempurl}


\bibitem[\protect\citeauthoryear{Schmidt, Pfister, and Zdeborov{\'a}}{Schmidt
  et~al\mbox{.}}{2019}]%
        {schmidt2019minimal}
\bibfield{author}{\bibinfo{person}{Christian Schmidt}, \bibinfo{person}{Henry~D
  Pfister}, {and} \bibinfo{person}{Lenka Zdeborov{\'a}}.}
  \bibinfo{year}{2019}\natexlab{}.
\newblock \showarticletitle{Minimal sets to destroy the k-core in random
  networks}.
\newblock \bibinfo{journal}{\emph{Physical Review E}} \bibinfo{volume}{99},
  \bibinfo{number}{2} (\bibinfo{year}{2019}), \bibinfo{pages}{022310}.
\newblock


\bibitem[\protect\citeauthoryear{Schwab, Bruinsma, Feldman, and Levine}{Schwab
  et~al\mbox{.}}{2010}]%
        {PhysRevE.82.051911}
\bibfield{author}{\bibinfo{person}{David~J. Schwab}, \bibinfo{person}{Robijn~F.
  Bruinsma}, \bibinfo{person}{Jack~L. Feldman}, {and} \bibinfo{person}{Alex~J.
  Levine}.} \bibinfo{year}{2010}\natexlab{}.
\newblock \showarticletitle{Rhythmogenic neuronal networks, emergent leaders,
  and $k$-cores}.
\newblock \bibinfo{journal}{\emph{Phys. Rev. E}}  \bibinfo{volume}{82}
  (\bibinfo{date}{Nov} \bibinfo{year}{2010}), \bibinfo{pages}{051911}.
\newblock
Issue 5.
\urldef\tempurl%
\url{https://doi.org/10.1103/PhysRevE.82.051911}
\showDOI{\tempurl}


\bibitem[\protect\citeauthoryear{Seidman}{Seidman}{1983}]%
        {seidman1983network}
\bibfield{author}{\bibinfo{person}{Stephen~B Seidman}.}
  \bibinfo{year}{1983}\natexlab{}.
\newblock \showarticletitle{Network structure and minimum degree}.
\newblock \bibinfo{journal}{\emph{Social networks}} \bibinfo{volume}{5},
  \bibinfo{number}{3} (\bibinfo{year}{1983}), \bibinfo{pages}{269--287}.
\newblock


\bibitem[\protect\citeauthoryear{Shan, Zhu, Xie, Wang, Zhou, Zhou, and
  Xuan}{Shan et~al\mbox{.}}{2021}]%
        {shan2021adversarial}
\bibfield{author}{\bibinfo{person}{Yalu Shan}, \bibinfo{person}{Junhao Zhu},
  \bibinfo{person}{Yunyi Xie}, \bibinfo{person}{Jinhuan Wang},
  \bibinfo{person}{Jiajun Zhou}, \bibinfo{person}{Bo Zhou}, {and}
  \bibinfo{person}{Qi Xuan}.} \bibinfo{year}{2021}\natexlab{}.
\newblock \showarticletitle{Adversarial Attacks on Graphs: How to Hide Your
  Structural Information}.
\newblock In \bibinfo{booktitle}{\emph{Graph Data Mining}}.
  \bibinfo{publisher}{Springer}, \bibinfo{address}{chi},
  \bibinfo{pages}{93--120}.
\newblock


\bibitem[\protect\citeauthoryear{Wen, Qin, Zhang, Lin, and Yu}{Wen
  et~al\mbox{.}}{2016}]%
        {wen2016efficient}
\bibfield{author}{\bibinfo{person}{Dong Wen}, \bibinfo{person}{Lu Qin},
  \bibinfo{person}{Ying Zhang}, \bibinfo{person}{Xuemin Lin}, {and}
  \bibinfo{person}{Jeffrey~Xu Yu}.} \bibinfo{year}{2016}\natexlab{}.
\newblock \showarticletitle{I/O efficient core graph decomposition at web
  scale}. In \bibinfo{booktitle}{\emph{2016 IEEE 32nd International Conference
  on Data Engineering (ICDE)}}. IEEE, \bibinfo{publisher}{IEEE},
  \bibinfo{address}{chi}, \bibinfo{pages}{133--144}.
\newblock


\bibitem[\protect\citeauthoryear{Wuchty and Almaas}{Wuchty and Almaas}{2005}]%
        {wuchty2005peeling}
\bibfield{author}{\bibinfo{person}{Stefan Wuchty} {and} \bibinfo{person}{Eivind
  Almaas}.} \bibinfo{year}{2005}\natexlab{}.
\newblock \showarticletitle{Peeling the yeast protein network}.
\newblock \bibinfo{journal}{\emph{Proteomics}} \bibinfo{volume}{5},
  \bibinfo{number}{2} (\bibinfo{year}{2005}), \bibinfo{pages}{444--449}.
\newblock


\bibitem[\protect\citeauthoryear{Xuan, Ruan, and Min}{Xuan
  et~al\mbox{.}}{2021}]%
        {xuan2021graph}
\bibfield{author}{\bibinfo{person}{Qi Xuan}, \bibinfo{person}{Zhongyuan Ruan},
  {and} \bibinfo{person}{Yong Min}.} \bibinfo{year}{2021}\natexlab{}.
\newblock \bibinfo{booktitle}{\emph{Graph Data Mining: Algorithm, Security and
  Application}}.
\newblock \bibinfo{publisher}{Springer Nature}, \bibinfo{address}{chi}.
\newblock


\bibitem[\protect\citeauthoryear{Zhang, Zhang, Qin, Zhang, and Lin}{Zhang
  et~al\mbox{.}}{2017}]%
        {zhang2017finding}
\bibfield{author}{\bibinfo{person}{Fan Zhang}, \bibinfo{person}{Ying Zhang},
  \bibinfo{person}{Lu Qin}, \bibinfo{person}{Wenjie Zhang}, {and}
  \bibinfo{person}{Xuemin Lin}.} \bibinfo{year}{2017}\natexlab{}.
\newblock \showarticletitle{Finding critical users for social network
  engagement: The collapsed k-core problem}. In
  \bibinfo{booktitle}{\emph{Thirty-First AAAI Conference on Artificial
  Intelligence}}. \bibinfo{publisher}{AAAI}, \bibinfo{address}{chi},
  \bibinfo{pages}{460--469}.
\newblock


\bibitem[\protect\citeauthoryear{Zhang, Cui, Neumann, and Chen}{Zhang
  et~al\mbox{.}}{2018}]%
        {zhang2018end}
\bibfield{author}{\bibinfo{person}{Muhan Zhang}, \bibinfo{person}{Zhicheng
  Cui}, \bibinfo{person}{Marion Neumann}, {and} \bibinfo{person}{Yixin Chen}.}
  \bibinfo{year}{2018}\natexlab{}.
\newblock \showarticletitle{An end-to-end deep learning architecture for graph
  classification}. In \bibinfo{booktitle}{\emph{Thirty-Second AAAI Conference
  on Artificial Intelligence}}. \bibinfo{publisher}{AAAI},
  \bibinfo{address}{chi}, \bibinfo{pages}{75--174}.
\newblock


\bibitem[\protect\citeauthoryear{Zhou, Lv, Mao, Wang, Yu, and Xuan}{Zhou
  et~al\mbox{.}}{2021}]%
        {zhou2021robustness}
\bibfield{author}{\bibinfo{person}{Bo Zhou}, \bibinfo{person}{Yuqian Lv},
  \bibinfo{person}{Yongchao Mao}, \bibinfo{person}{Jinhuan Wang},
  \bibinfo{person}{Shanqing Yu}, {and} \bibinfo{person}{Qi Xuan}.}
  \bibinfo{year}{2021}\natexlab{}.
\newblock \showarticletitle{The Robustness of Graph k-shell Structure under
  Adversarial Attacks}.
\newblock \bibinfo{journal}{\emph{IEEE Transactions on Circuits and Systems II:
  Express Briefs}} \bibinfo{volume}{5}, \bibinfo{number}{3}
  (\bibinfo{year}{2021}), \bibinfo{pages}{23--28}.
\newblock


\end{thebibliography}

\bibliographystyle{ACM-Reference-Format}

\end{document}